\documentclass[a4paper,UKenglish,cleveref, autoref, thm-restate]{lipics-v2021}
\usepackage{mathtools}
\def\set#1{\mathord{\{#1\}}}

\title{Optimal Repairs for Unary Functional Dependencies: Resolving the Case of Updates}

\author{Benny Kimelfeld}{Technion -- Israel Institute of Technology}{bennyk@cs.technion.ac.il}{https://orcid.org/0000-0002-7156-1572}{}

\author{Ester Livshits}{Technion -- Israel Institute of Technology}{esterliv@technion.ac.il}{https://orcid.org/0000-0003-3485-9887}{}

\authorrunning{Kimelfeld and Livhits}

\Copyright{CC-BY} 



\category{} 

\relatedversion{} 

\EventEditors{John Q. Open and Joan R. Access}
\EventNoEds{2}
\EventLongTitle{42nd Conference on Very Important Topics (CVIT 2016)}
\EventShortTitle{CVIT 2016}
\EventAcronym{CVIT}
\EventYear{2016}
\EventDate{December 24--27, 2016}
\EventLocation{Little Whinging, United Kingdom}
\EventLogo{}
\SeriesVolume{42}
\ArticleNo{23}

\usepackage{algorithm,algorithmicx}
\usepackage[noend]{algpseudocode}
\usepackage{xspace}
\usepackage{caption}
\usepackage{stackengine}
\usepackage{paralist}
\usepackage{color}
\usepackage{tabularx}
\usepackage{tikz}
\usetikzlibrary{positioning, fit,shapes.multipart,shapes.geometric}

\usepackage{cite}

\usepackage{enumitem}

\usepackage{ifmtarg}
\makeatletter
\newcommand{\algcaption}[2][]{%
  \refstepcounter{algorithm}%
  \@ifmtarg{#1}
    {\addcontentsline{loa}{figure}{\protect\numberline{\thealgorithm}{\ignorespaces #2}}}
    {\addcontentsline{loa}{figure}{\protect\numberline{\thealgorithm}{\ignorespaces #1}}}%
  \toprule
  \textbf{\fname@algorithm~\thealgorithm}\ #2\par 
  \midrule
}
\makeatother

\def\vals{\mathbf{Val}}
\def\atts{\mathbf{Att}}
\def\tids{\mathsf{tids}}

\newcommand{\defeq}{\vcentcolon=}
\newcommand{\eqdef}{=\vcentcolon}

\def\e#1{\emph{#1}}

\def\ra{\rightarrow}
\def\la{\leftarrow}
\def\lra{\leftrightarrow}
\newcommand{\eat}[1]{}

\def\att#1{\textsf{#1}}

\def\closure#1#2{#1^+_{#2}}

\def\e#1{\emph{#1}}

\def\eqdef{\stackrel{\textsf{\tiny def}}{=}}

\newcommand{\algname}[1]{{\sf #1}}
\def\myrulewidth{3.20in}

\def\therule{\rule{\myrulewidth}{0.2pt}}

\newenvironment{insidecode}[3]
{
\begin{tabular}{p{\myrulewidth}}
\multicolumn{1}{c}{\rule{0mm}{3mm}{\bf #3} $\algname{#1}(\mbox{#2})$\vspace{-0.6em}}\\
\therule\vskip-0.8em\therule
\vspace{0em}
\begin{algorithmic}[1]}
{\end{algorithmic}
\vskip-0.3em\therule
\end{tabular}}

\newcommand{\graph}{\mathcal{G}}

\newcommand{\depset}{\mathrm{\Delta}}

\def\phi{\varphi}

\newcommand{\comptask}[3]{
\vskip1em
\fbox{
\begin{tabular}{rp{0.3\textwidth}}
  \textbf{Problem:} & \e{#1}\\\hline
  \textbf{Input:} & #2\\
  \textbf{Goal:} & #3
\end{tabular}
}
\vskip1em
}

\newcommand{\remarkend}{\hfill$\triangleleft$}

\newenvironment{repeatresult}[2]
{\vskip0.5em\par\textsc{#1} #2.\em}
{\vskip1em}

\newenvironment{repproposition}[1]{\begin{repeatresult}{Proposition}{#1}}{\end{repeatresult}}
\newenvironment{reptheorem}[1]{\begin{repeatresult}{Theorem}{#1}}{\end{repeatresult}}

\newenvironment{proofsketch}{\begin{proof}{\textit{(Sketch)}\,\,}}
{\end{proof}}

\newtheorem{commentthm}[theorem]{Comment}

\def\true{\mathbf{true}}

\def\Ptime{\mathrm{P}}
\def\NP{\mathrm{NP}}

\def\partitle#1{\vskip0.3em\par\noindent\textbf{#1.}\,\,}

\def\att#1{\textsf{#1}}

\def\nodes{V}
\def\edges{E}

\newcommand{\mathsc}[1]{{\normalfont\textsc{#1}}}

\def\ra{\rightarrow}
\def\la{\leftarrow}
\def\lra{\leftrightarrow}

\def\distu{\mathit{dist}}

\newenvironment{subroutine}
{\begin{algorithm}\floatname{algorithm}{Subroutine}}
{\end{algorithm}}

\newcounter{subroutine}

\newcommand{\cut}[1]{}

\newcommand{\attr}{{\mathit{attr}}}

\definecolor{LightCyan}{rgb}{0.88,1,1}
\definecolor{Gray}{gray}{0.9}

\newcommand{\benny}[1]{{\texttt {\color{magenta} Benny: [{#1}]}}}
\newcommand{\ester}[1]{{{\tt \color{purple} Ester: [{#1}]}}}

\newcommand{\nat}{{\mathbb{N}}}

\newcommand{\proj}[1]{{\Pi}}
\newcommand{\sel}[1]{{\sigma}}

\newcommand{\cutfull}[1]{}

\newcommand{\commentresolved}[1]{}

\newcommand{\setof}[2]{\{{#1}\mid{#2}\}}        
\cut{

}

\renewcommand{\comptask}[3]{
\vskip1em\noindent
\fbox{%
\begin{tabularx}{\dimexpr\linewidth-2\fboxsep-2\fboxrule\relax}{rX}
  \textbf{Problem:} & \emph{#1} \\ \hline
  \textbf{Input:} & #2 \\
  \textbf{Goal:} & #3
\end{tabularx}%
}
\vskip1em
}

\def\scdepset#1{\depset^{\mathsf{sc}}_{#1}}

\begin{document}
\nolinenumbers

\maketitle

\begin{abstract}
If a table violates its required set of functional dependencies (FDs), what is the minimum number of cell changes needed to restore consistency? This fundamental problem, known as finding an optimal update repair (U-repair), is known to admit polynomial-time algorithms only for a small number of specific FD sets. Whether additional tractable cases exist has remained open. The only established hardness result for this problem is due to Kolahi and Lakshmanan (2009); subsequent attempts to prove hardness for additional cases have failed, leaving these cases unresolved. In this work, we make substantial progress on this open problem by completely resolving the case of unary FDs, in which every FD has a single attribute on its left-hand side. We show that every set of unary FDs either falls into one of the previously known tractable classes or makes the problem of finding an optimal U-repair NP-hard.
\end{abstract}

\section{Introduction}
\label{sec:introduction}

Suppose that a database violates a set of functional dependencies (FDs). \e{How many cell values need to be changed, at minimum, to restore consistency?} This question defines the problem of computing an \e{optimal update repair} (or \e{optimal U-repair}): we seek a database that satisfies the FDs and differs from the original database in as few cells as possible. Despite the elementary formulation of the problem and the central role of FDs in databases, its computational complexity is surprisingly poorly understood. As we explain later, polynomial-time algorithms are known only for a few restricted classes of FD sets, and the only established hardness result applies to a specific configuration of two FDs. Whether further tractable cases exist has remained open.

This problem is a natural instance of the general framework of database repairs. Inconsistencies arise in databases for many reasons. Data may originate from noisy or imprecise sources, such as sensors, social platforms, or automatically extracted content, and may be produced by error-prone processes such as machine learning and generative Artificial Intelligence (AI). Inconsistencies can also emerge when integrating databases from different sources, which may provide conflicting information or represent the same information in incompatible ways. Arenas et al.~\cite{DBLP:conf/pods/ArenasBC99} introduced a principled approach to managing such inconsistencies through the notions of \e{repairs} and \e{consistent query answering}. Given a database $D$ that violates a collection of integrity constraints, a repair is a consistent database $D'$ obtained from $D$ through a minimal collection of editing operations. Consistent query answering then seeks the query answers that hold across all repairs.

The repair framework admits many variants, according to three basic choices: the \e{integrity constraints} that define consistency, the \e{operations} allowed for restoring consistency, and the criterion used to measure minimality~\cite{DBLP:conf/icdt/AfratiK09}. Integrity constraints commonly include denial constraints~\cite{DBLP:journals/jiis/GaasterlandGM92}, and in particular the classical functional dependencies, as well as inclusion dependencies~\cite{DBLP:journals/jcss/CasanovaFP84}, which include referential constraints. Repair operations may consist of deleting or inserting tuples, or modifying individual cell values. Finally, minimality may be \e{local}, requiring that no proper subset of the performed operations suffices, or \e{global}, requiring minimum total cost among all ways of restoring consistency. For instance, when repairs are restricted to tuple deletions, a \e{subset repair}~\cite{DBLP:journals/iandc/ChomickiM05} is locally minimal, whereas a \e{cardinality repair}~\cite{DBLP:conf/icdt/LopatenkoB07} minimizes the number of deleted tuples. Operations may also be assigned nonuniform costs, for example to reflect different levels of confidence in data items~\cite{DBLP:conf/icdt/LopatenkoB07,DBLP:conf/icdt/KolahiL09,DBLP:conf/icdt/GiladIK23} or distance between the original and the updated values~\cite{DBLP:conf/icdt/KaminskyKLNW25,DBLP:conf/icde/KoudasSSV09}.

We study the natural combination in which consistency is defined by FDs, repairs are obtained exclusively through value updates, and minimality is global. In the basic setting, changing a single cell has unit cost, so the objective is precisely to minimize the number of changed cells. Our results extend directly to weighted updates, where cells may have different costs. This notion was called an ``optimum V-repair'' by Kolahi and Lakshmanan~\cite{DBLP:conf/icdt/KolahiL09}, and we adapt the term ``optimal U-repair'' by Livshits et al.~\cite{DBLP:journals/tods/LivshitsKR20}.

The current understanding of this optimization problem leaves a substantial gap between tractability and hardness. Livshits et al.~\cite{DBLP:journals/tods/LivshitsKR20} established polynomial-time algorithms for two general forms of FD sets, up to equivalence. The first is an \e{lhs-chain}, where the left-hand sides of the FDs form a chain under inclusion. The second is a \e{matching constraint} $\set{A\ra B,B\ra A}$, which we denote by $A\lra B$. When all FDs are unary, these results provide tractability essentially only when the FD set reduces to a single determinant $\set{A\ra X}$ or to a matching constraint $A\lra B$. In addition, when an FD set is a union of FD sets that share no common attributes, then the two corresponding repair problems can be solved separately (hence, $\set{A\lra B,C\lra D}$ is also tractable). 

The known hardness landscape is even more limited. Kolahi and Lakshmanan~\cite{DBLP:conf/icdt/KolahiL09} proved NP-hardness for $\set{A\ra B,B\ra C}$. Attempts to extend hardness beyond this case have not succeeded. In particular, in their conference version, Livshits et al.~\cite{DBLP:conf/pods/LivshitsKR18} claimed hardness for $\set{A\lra B\ra C}$; an error in the proof was subsequently discovered, as they discuss in the journal version~\cite{DBLP:journals/tods/LivshitsKR20}. More recently, Miao et al.~\cite{DBLP:journals/vldb/MiaoZLWC23} claimed NP-hardness for $\set{A\ra C,B\ra C}$. As we explain in \Cref{sec:flaw-acbc}, the proposed reduction does not establish this claim.

Hence, even after restricting attention to \e{unary FDs} (i.e., FDs with a single attribute on the left-hand side), a basic question has remained unanswered: are the known polynomial-time cases isolated exceptions, or are there additional classes of unary FDs that admit efficient optimal U-repair algorithms? We answer this question completely. We prove a dichotomy for unary FDs: for every set $\depset$ of unary FDs, either $\depset$ falls, up to equivalence, within one of the previously known tractable cases, or computing an optimal U-repair with respect to $\depset$ is NP-hard. In particular, the previously known algorithms exhaust the tractable cases among unary FDs.

\section{Preliminaries}
\label{sec:preliminaries}

We begin with preliminary definitions and notation.

\subparagraph*{Relations.}
Throughout the paper, we will assume countably infinite sets $\vals$ of \e{values} and $\atts$ of \e{attributes}. A \e{relation schema} is a set $R=\set{A_1,\dots,A_k}$ of attributes. An \e{$R$-tuple} is a function $t:R\ra\vals$
that maps each attribute $A\in R$ to a value that we denote by $t[A]$. A \e{relation} $r$ is associated with a relation schema, denoted
$\atts(r)$, a finite set of tuple identifiers, denoted $\tids(r)$, and a mapping from $\tids(r)$ to 
$\atts(r)$-tuples. We say that $r$ is a relation \e{over} the relation schema $\atts(r)$. We denote by $r[i]$ the tuple that $r$ maps to the identifier $i$. Hence, $r[i][A]$ is the value that tuple $i$ has for the attribute $A$. 

\begin{remark}
Note that we allow duplicate tuples, since we do not assume that tuples with different identifiers are necessarily distinct.\remarkend
\end{remark}

Let $X$ be a set of attributes. We denote by $\pi_X r$ the projection of $r$ onto $X$. Formally, $\pi_X r$ is the relation $r'$ such that $\atts(r')=X$, $\tids(r')=\tids(r)$, and $r'[i][A]=r[i][A]$ for every $A\in\atts(r)\cap X$. Observe that, under our notation, $(\pi_X r)[i]$ is the projection of the tuple $r[i]$ onto $X$. As a shorthand, we write $r[i][X]$ instead of $(\pi_X r)[i]$.

\subparagraph*{Functional dependencies.}
A \e{functional dependency}, or \e{FD} for short, is an expression of the form $X\ra Y$, where $X$ and $Y$ are finite sets of attributes. We say that $X\ra Y$ is \e{over} a relation schema $R$ if $X\cup Y\subseteq R$.
A relation $r$ over a relation schema containing $X\cup Y$ satisfies the FD $X\ra Y$ if the values of the attributes in $X$ uniquely determine the values of the attributes in $Y$. That is, whenever two tuples of $r$ agree on all attributes of $X$, they must also agree on all attributes of $Y$. More formally, $r$ satisfies $X\ra Y$ if, for all tuple identifiers $i$ and $i'$ in $\tids(r)$, if $r[i][X]=r[i'][X]$ then $r[i][Y]=r[i'][Y]$.
For a set $\Delta$ of FDs over $\atts(r)$, we denote by $r\models \Delta$ the statement that $r$ satisfies every FD in $\Delta$.

An FD $X \ra Y$ is \e{unary} if $X$ consists of a single attribute, and it is \e{trivial} if $Y\subseteq X$, in which case it is satisfied by every relation. A \e{matching constraint} is a constraint of the form $X\lra Y$, which is shorthand for the pair $\set{X\ra Y, Y\ra X}$. An attribute $A$ is a
\e{common left-hand side} (common lhs) of $\depset$ if $A$ occurs in the
left-hand side of every FD in $\depset$.

In our examples, we adopt the standard convention of omitting curly braces and commas when writing sets of attributes. For example, we write $AB$ instead of $\set{A,B}$. We also use a compact notation for sets of FDs by combining multiple dependency expressions and possibly varying the direction of the arrows. For example,
$A\lra B\la CD$ is shorthand for the set 
$\set{A\ra B,B\ra A,CD\ra B}$.
The closure of a set $\Delta$ of FDs, denoted $\Delta^+$, is the set of all FDs implied by $\Delta$ (equivalently, all FDs that can be derived from $\Delta$ using Armstrong's axioms). In particular, $\Delta^+$ contains all trivial FDs. 

The closure of a finite set $X$ of attributes with respect to $\Delta$, denoted $\closure{X}{\Delta}$, is the set of all attributes $A$ such that $X\ra A$ belongs to $\Delta^+$. Two finite sets $X$ and $Y$ of attributes are \e{equivalent} (with respect to $\Delta$) if they have the same closure, that is, if $\closure{X}{\Delta}=\closure{Y}{\Delta}$ (or, equivalently, $X\ra Y$ and $Y\ra X$ both belong to $\Delta^+$). By a slight abuse of notation, we say that two attributes $A$ and $B$ are \e{equivalent} if the singleton sets $\set{A}$ and $\set{B}$ are equivalent.

Finally, for a set $\Delta$ of FDs, we denote by $\atts(\Delta)$ the set of all attributes that appear in either the left-hand or right-hand side of an FD in $\Delta$.

\subparagraph*{Strongly Connected Components (SCCs) of attributes.}
When $\Delta$ is a set of unary FDs, it is convenient to view it as a directed graph over its attributes. Since a relation satisfies $\Delta$ if and only if it satisfies its closure $\Delta^+$, the two induce exactly the same repair problem; we therefore define the graph over $\Delta^+$, which is more convenient to work with.

\begin{definition}[FD graph]\label{def:fd-graph}
Let $\Delta$ be a set of unary FDs over $R$. The \e{FD graph} of $\Delta$, denoted $\graph(\Delta)=(\nodes,\edges)$, is the directed graph whose node set $\nodes$ is the set $\atts(\Delta)$ of attributes occurring in $\Delta$, and in which $(A,B)\in\edges$ if and only if $A\neq B$ and $A\ra B\in\Delta^+$.
\end{definition}

The \e{strongly connected components} (SCCs) of $\Delta$ are the sets of FDs induced by the SCCs of $\graph(\Delta)$. Hence, the SCCs are the equivalence classes of $\atts(\Delta)$ under the aforementioned equivalence relation on attributes: two attributes lie in the same SCC if and only if they are equivalent, that is, each determines the other.
Since $\graph(\Delta)$ is transitively closed, the SCCs are the maximal complete directed subgraphs.

The \e{weakly connected components} of $\Delta$ are the sets of FDs induced by the weakly connected components of $\graph(\Delta)$; that is, two FDs lie in the same weakly connected component whenever their attributes are joined by a path in the undirected graph underlying $\graph(\Delta)$. We say that $\Delta$ is \e{weakly connected} if $\graph(\Delta)$ consists of a single weakly connected component.

\section{Optimal Update Repairs}\label{sec:main-result}

In this section, we formally introduce the computational problem studied in
this paper, that of computing an \e{optimal update repair} of an inconsistent
relation. We also give an overview of the state of affairs for this problem.

\subsection{Problem Definition}

As is standard in the study of
repairs, we measure complexity in the sense of \e{data complexity}: the
relation schema $R$ and the set $\depset$ of FDs are \e{fixed}, and the input
consists of a single relation $r$ over $R$. 

\subparagraph*{Updates and optimal U-repairs.}
Let $r$ be a relation over a schema $R$. An \e{update} of $r$ is a relation
$r'$ over $R$ obtained from $r$ by modifying attribute values, without
inserting or deleting tuples; formally, $r'$ is an update of $r$ if $\tids(r')=\tids(r)$. Note that the values used in an
update are not restricted to the \e{active domain} of $r$ (the values already
occurring in $r$); any value in $\vals$ may be used. Following prior work~\cite{DBLP:conf/icdt/KolahiL09,DBLP:journals/tods/LivshitsKR20}, we define the \e{distance} from an
update $r'$ to $r$, denoted $\distu(r',r)$, as the total number of cells in
which $r'$ differs from $r$:
$$
  \distu(r',r)\;\eqdef\;
  \sum_{i\in\tids(r)}
  \bigl|\set{A\in\atts(r)\mid r'[i][A]\neq r[i][A]}\bigr|\,.
$$
\begin{remark}
We may also wish to capture scenarios in which modifying certain cells is more costly than modifying others. 
To this end, the definition of distance naturally extends to a \emph{weighted} distance, in which each cell is assigned a specific update cost. For clarity of presentation, we focus on the unweighted distance defined above. However, our results immediately extend to this weighted case, as we explain in \Cref{sec:conclusions}. \remarkend
\end{remark}

Fix a set $\depset$ of FDs over $R$. A \e{consistent update} of $r$ (with
respect to $\depset$) is an update $r'$ of $r$ such that $r'\models\depset$.

\begin{definition}[Optimal update repair]\label{def:opt-u-repair}
An \e{optimal update repair} of $r$ (with respect to $\depset$), or \e{optimal
U-repair} for short, is a consistent update $r'$ of $r$ that minimizes
$\distu(r',r)$ among all consistent updates of $r$; that is,
$\distu(r',r)\le\distu(r'',r)$ for every update $r''$ of $r$ with
$r''\models\depset$.
\end{definition}

The computational problem we study is that of producing such an
optimum.

\comptask{Computing an optimal U-repair for $(R,\depset)$}
{A relation $r$ over $R$.}
{Compute an optimal U-repair of $r$ with respect to $\depset$.}

%
Note that every pair $(R,\depset)$ defines a separate computational problem. A straightforward observation is that $R$ does not play any role in the complexity of this problem, as long as it contains all of the attributes in $\atts(\Delta)$. Our main result, stated formally as Theorem~\ref{thm:main_dichotomy} in Section~\ref{sec:classification}, is a complete dichotomy in the complexity of this problem for the case in which $\depset$ consists of unary FDs. Specifically, we classify all such $\depset$ into ones where an optimal U-repair can be found in polynomial time, and ones that induce an NP-hard optimization problem. To present this result, we first need to review the state of affairs for this problem.

\subsection{State of Affairs}\label{subsec:state-of-affairs}

The complexity of computing optimal repairs by value updates was first studied
by Kolahi and Lakshmanan~\cite{DBLP:conf/icdt/KolahiL09} and later, in a systematic
way, by Livshits et al.~\cite{DBLP:journals/tods/LivshitsKR20}. The
latter work develops a rich toolbox for the problem, yet the exact tractability
frontier for update repairs has remained elusive; indeed, in contrast to the
full dichotomy they establish for \e{subset} repairs, the authors explicitly
leave the complexity of optimal \e{update} repairs largely open. We summarize
below what is known, focusing on unary FDs.

\subparagraph*{Decomposition.}
Livshits et al.~\cite{DBLP:journals/tods/LivshitsKR20} establish a
\e{decomposition} property: the optimal U-repair problem for $\depset$ can be
solved by partitioning $\depset$ into its weakly connected components,
computing an optimal U-repair for each component independently, and combining
the resulting cell updates, which are guaranteed to touch pairwise disjoint
sets of attributes. Consequently, the problem for $\depset$ is solvable in
polynomial time whenever it is solvable in polynomial time for every weakly
connected component of $\depset$; conversely, if the problem is NP-hard for even
a single weakly connected component, then it is NP-hard for $\depset$ as a
whole. 

\subparagraph*{Two tractable cases.}
Livshits et al.~\cite{DBLP:journals/tods/LivshitsKR20} further establish the
tractability of two families of FD sets. The first consists of the \e{chain} FD
sets: $\depset$ is a \e{chain} if its FDs can be arranged so that the left-hand
side of each is contained in the left-hand side of the next. For \e{unary} FDs
this criterion is restrictive: since every left-hand side is a single attribute,
all FDs of a chain must share the \e{same} left-hand-side attribute, so the case
collapses to $\depset$ having a \e{common lhs}. The second family is the
matching $\set{A\ra B,\,B\ra A}$. Combined with the aforementioned decomposition property, this yields the following
state of knowledge on the positive side: the optimal U-repair problem is in
$\Ptime$ if every weakly connected component of $\depset$ has a common lhs or
is a matching. 

\subparagraph*{Known intractable case.}
On the negative side, the only FD set previously known to induce an NP-hard
optimal U-repair problem is the directed path $\set{A\ra B,\,B\ra C}$, whose
hardness was established by Kolahi and Lakshmanan~\cite{DBLP:journals/tods/LivshitsKR20}. This single
intractable component, together with the decomposition property, is the extent
of the known hardness landscape.

Two attempts to enlarge the map of intractable cases have appeared, both of
which turned out to be flawed. First, Livshits et
al.~\cite{DBLP:conf/pods/LivshitsKR18} provided an NP-hardness argument for
the FD set $\set{A\ra B,\,B\ra A,\,B\ra C}$, but subsequently acknowledged an error in that proof~\cite{DBLP:journals/tods/LivshitsKR20}.
Second, Miao et al.~\cite{DBLP:journals/vldb/MiaoZLWC23} claimed NP-hardness for the FD set $\set{A\ra C,\,B\ra C}$; as we discuss in
Section~\ref{sec:flaw-acbc}, their proof contains a flaw.

\subsection{Main Result}\label{sec:classification}
We make a significant step toward resolving the long-standing open problem of the complexity of U-repairs: we
provide a \e{complete} complexity classification of the optimal U-repair
problem for the case of unary FDs. As it turns out, the two polynomial-time cases identified by Livshits et
al.~\cite{DBLP:journals/tods/LivshitsKR20} are, in fact, \e{the only} tractable
cases (assuming $\Ptime\neq\NP$): for unary FDs, every FD set that is not
captured by these two cases gives rise to an NP-hard optimal U-repair problem. Note that trivial FDs never take part
in a violation and can be removed from $\depset$ without changing the set of
consistent updates of any relation; therefore, we assume throughout that
$\depset$ contains no trivial FDs.

\begin{theorem}\label{thm:main_dichotomy}
Let $\depset$ be a set of nontrivial unary
FDs. An optimal U-repair can be computed in polynomial time if every weakly connected component $\depset'$ of $\depset$ has a common lhs ($A\ra X$) or
is a matching constraint ($A\lra B$). Otherwise, computing an optimal U-repair is
NP-hard.
\end{theorem}


As aforesaid, the tractable side of Theorem~\ref{thm:main_dichotomy}, the case in which every
weakly connected component has a common lhs or is a matching, has already been
settled by Livshits et al.~\cite{DBLP:journals/tods/LivshitsKR20}. Our
contribution is therefore the hardness side, which we establish for the natural
decision variant of the problem: given a relation $r$ and a budget
$K\in\nat$, decide whether there is a consistent update $r'$ of $r$ with
$\distu(r',r)\le K$. 

\subsection{Proof Strategy}
We prove the hardness side of \Cref{thm:main_dichotomy} using two main steps. In the first step, we establish the NP-hardness for several specific cases. 
One case is where the FD set is a single SCC with three or more attributes (\Cref{sec:large-scc}). 
Other cases are specific FD sets involving three or four attributes (\Cref{sec:small-sccs}). 
Among the FD sets that we handle are
those whose complexity has been left open following published proofs that were found flawed, as we discussed in \Cref{subsec:state-of-affairs}; hence, this part of the proof resolves these open problems.
The second part (\Cref{sec:scc-reduction}) extends the specific cases of the first part to the general case of a set of unary FDs. We first show that whenever the FD set does not belong to the positive side of \Cref{thm:main_dichotomy}, there is a set of attributes, called \e{convex}, that induces one of the special cases proved hard in the first part. Finally, we show a reduction from the U-repair problem for any convex set of attributes to the original problem (involving all attributes).

\section{Three or More Strongly Connected Attributes}\label{sec:large-scc}

We begin our analysis with the case of a strongly connected set of FDs (consisting of a single SCC). Recall that $\depset$ is \e{strongly connected} if it is a set of pairwise-equivalent attributes. We denote by $\scdepset{c}$ the strongly connected FD set with $c$ attributes. (Note that all such FD sets are equivalent.) Specifically, we assume that 
\[
\scdepset{c} \defeq \set{A_i\ra A_j \mid i,j\in\set{1,\dots,c}  \land i\neq j}
\mbox{\quad with \quad} \atts(\scdepset{c})=\set{A_1,\dots,A_c}
\,.
\]

\begin{theorem}\label{thm:kc-hard}
For every $c \ge 3$, computing an optimal U-repair for $\scdepset{c}$ is NP-hard.
\end{theorem}
The proof of \Cref{thm:kc-hard} is nontrivial, and 
the rest of this section sketches the proof. The complete construction and formal analysis are provided in the Appendix. Importantly, the combinatorial core of the proof lies in the case of three attributes. We establish this case by a reduction from a restricted variant of three-dimensional matching, and we then reduce it to $\scdepset{c}$ for every $c>3$ by padding each tuple with values that no repair within the budget can preserve. Hence, the challenge of the proof is concentrated on the case of $c=3$, and it is the case that we sketch here.

The fact that all attributes in $\scdepset{c}$ determine one another imposes a strict structure on consistent relations. If two tuples of a relation $r \models \scdepset{c}$ agree on even a single attribute $A_i$, then they must agree on all attributes $A_1,\dots,A_c$. Equivalently, every two tuples are either completely identical on $A_1,\dots,A_c$, or they disagree on every single one of them.  Consequently, a consistent relation partitions its tuples into clusters of identical tuples. Each cluster is defined by a single \e{target tuple}
$v\in\vals^c$ shared by all its members, and the assignments of any two distinct clusters must disagree in all $c$ coordinates. Repairing $r$ is thus the act of assigning a target tuple $v$ to every tuple: a tuple $t$ modified to match $v$ is charged one unit for each coordinate on which $t$ and $v$ 
differ (that is, $c$ minus the number of original cells it retains). Computing an optimal U-repair for $\scdepset{c}$ is therefore the problem of partitioning the tuples into these disjoint clusters and modifying each cluster to a joint $v$, which is unique on every attribute, while preserving as many original cells as possible.

Restricted to three attributes, this clustering problem is closely
related to the \e{3-Dimensional Matching} 
(3DM) problem. Consider the values that occur in these three attributes as elements of three disjoint sets $X$, $Y$, and
$Z$, so that every tuple is a triple in $X\times Y\times Z$. In 3DM, one is given such sets, each of size $q$, and a set $S\subseteq X\times Y\times Z$ of triples, and the goal is to select a subset of $S$ that touches every element of $X\cup Y\cup Z$ exactly once; equivalently, we ask whether we can find $q$ pairwise-disjoint triples. Our repair problem has a similar shape: a repair, too, selects a set of target triples that must be pairwise disjoint, and it pays nothing for a tuple it leaves unchanged and one unit for each cell it modifies. 
However, in 3DM, the chosen triples must belong
to $S$, whereas a repair may use \e{any} triple of $\vals^3$ as a target triple. 
To this aim, we define a restricted form of 3DM.

\subparagraph*{Restricted matching problem (R3DM).}
The restricted variant we reduce from, R3DM, imposes two structural
conditions on the input triples of 3DM:
\begin{itemize}
  \item  Limited Intersection (LI): every two distinct triples share at most one element;
  \item Triangle-freeness (TF): there are no three distinct triples and
        three elements $x,y,z$ such that one triple contains
        $\set{x,y}$, another $\set{x,z}$, and the third
        $\set{y,z}$.
\end{itemize}
We prove in the Appendix that R3DM is still NP-complete.
As we explain later, our reduction relies heavily on the LI and TF properties for correctness, specifically to show that an optimal U-repair yields a perfect matching.


\subparagraph*{Reduction gadget.}

\begin{figure}[t]
\centering
\begin{tikzpicture}[
    cell/.style={rectangle split, rectangle split parts=3,
                 rectangle split horizontal, draw=black!75, semithick,
                 rounded corners=1pt,
                 minimum height=5.5mm, text width=9mm, align=center,
                 font=\small, inner sep=1.5pt, fill=white},
    rowname/.style={font=\small, anchor=east},
    colname/.style={font=\footnotesize\itshape, text=black!70},
    blockname/.style={font=\footnotesize\itshape, text=black!70},
    ann/.style={font=\footnotesize, anchor=west, text=black!70},
    anchorpin/.style={->, dashed, semithick, red!65!black}
  ]
  \def\el#1{\textcolor{blue!55!black}{#1}}
  \def\fc#1{\textcolor{red!70!black}{#1}}
  \def\fl#1{\textcolor{black!70}{#1}}

  \node[cell] (t) at (0,0)
    {$\el{x}$ \nodepart{two} $\el{y}$ \nodepart{three} $\el{z}$};
  \node[cell, below=1.6mm of t] (u1)
    {$\fc{c_1}$ \nodepart{two} $\el{y}$ \nodepart{three} $\el{z}$};
  \node[cell, below=1.6mm of u1] (u2)
    {$\el{x}$ \nodepart{two} $\fc{c_2}$ \nodepart{three} $\el{z}$};
  \node[cell, below=1.6mm of u2] (u3)
    {$\el{x}$ \nodepart{two} $\el{y}$ \nodepart{three} $\fc{c_3}$};

  \node[rowname, left=1.2mm of t]  {$t_s$};
  \node[rowname, left=1.2mm of u1] {$u^1_s$};
  \node[rowname, left=1.2mm of u2] {$u^2_s$};
  \node[rowname, left=1.2mm of u3] {$u^3_s$};

  \node[colname, above=1mm of t.one north]   {$A_1$};
  \node[colname, above=1mm of t.two north]   {$A_2$};
  \node[colname, above=1mm of t.three north] {$A_3$};

  \node[cell, right=18mm of u1] (e1)
    {$\fc{c_1}$ \nodepart{two} $\fl{d_{1,2}}$ \nodepart{three} $\fl{d_{1,3}}$};
  \node[cell, below=1.6mm of e1] (e2)
    {$\fl{d_{2,1}}$ \nodepart{two} $\fc{c_2}$ \nodepart{three} $\fl{d_{2,3}}$};
  \node[cell, below=1.6mm of e2] (e3)
    {$\fl{d_{3,1}}$ \nodepart{two} $\fl{d_{3,2}}$ \nodepart{three} $\fc{c_3}$};

  \node[ann, right=1.2mm of e1] {$\times N$};
  \node[ann, right=1.2mm of e2] {$\times N$};
  \node[ann, right=1.2mm of e3] {$\times N$};

  \draw[anchorpin] (u1.east) -- (e1.west);
  \draw[anchorpin] (u2.east) -- (e2.west);
  \draw[anchorpin] (u3.east) -- (e3.west);

  \node[blockname, above=6mm of t] (gtitle) {Gadget of $s=(x,y,z)$};
  \node[blockname] at (e1 |- gtitle) {Enforcing groups};
\end{tikzpicture}
\caption{The tuples that the reduction introduces for a triple $s=(x,y,z)$. Each
auxiliary tuple inherits two values from the $S$-tuple and contains a fresh constant
in the remaining attribute; the dashed arrows lead to the enforcing group that anchors
this constant.}
\label{fig:construction}
\end{figure}
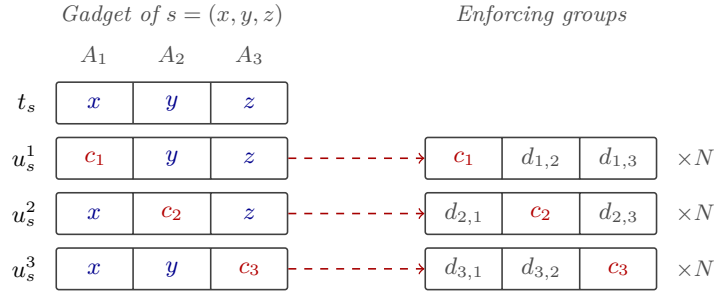

Next, we describe our reduction, which is illustrated in Figure~\ref{fig:construction}. Fix an \mathsc{R3DM} instance. For each triple $s=(x,y,z)$ in the instance, we build a gadget on the three attributes, as shown in \Cref{fig:construction}.
\begin{itemize}
  \item An \e{$S$-tuple} $(x,y,z)$, that is, $s$ itself;
  \item Three \e{auxiliary tuples}: the auxiliary tuple for a given attribute agrees with the $S$-tuple on the other two attributes, but contains a fresh constant in its own attribute.
\end{itemize}
We refer to these four tuples as the \e{gadget tuples}.
Finally, to prevent fresh constants from being chosen as part of a target triple, each fresh constant is anchored by an \e{enforcing group}: a set of $N$ identical tuples that agree with the auxiliary tuple on the fresh constant and carry, in every other attribute, a value that occurs nowhere else. We set the budget to $K\defeq8m-5q$ and the size of the enforcing groups to $N\defeq K+1$. All in all, the constructed relation $r$ consists of the $4m$ tuples of the $m$ gadgets and the $3m$ enforcing groups.

\subparagraph*{Correctness}
We first show that a perfect matching yields a U-repair.
This direction is fairly simple. Given a perfect matching $M\subseteq S$, we use the triples of $M$ as the target tuples, and update every gadget tuple to be equal to a target tuple that is closest to it. (In general, finding a low-cost U-repair for 
$\scdepset{c}$ boils down to finding a good set of target tuples.) 

Next, we show that an in-budget U-repair yields a perfect matching.
This direction is considerably more challenging. Consider a consistent update $r'$ of cost at most $K$. Every two target tuples are disjoint. We show that $r'$ has at least $q$ clusters where each has a target tuple that is an $S$-tuple (i.e., a triple of the R3DM instance). Hence, these $q$ target tuples form a perfect matching.

To this end, we first show that the enforcing groups prevent the gadget tuples from retaining their fresh constants. Moreover, $r'$ has no reason to change the enforcing tuples. Therefore, the cost of $r'$ is the number of cells changed within the gadget tuples. We also show that there is no reason to update more than two cells in each tuple of $r$. Hence, the cost $c_t$ incurred by each individual gadget tuple $t$ is at most two. For convenience, we refer to $2-c_t$ as the \e{saving} of $t$. We further refer to the \e{saving of a cluster} as the total savings of all tuples in the cluster. With zero savings, we have $4m$ tuples with a total cost of $8m$. Since the total cost of $r'$ is at most $K=8m-5q$, we conclude that the total savings over all clusters are at least $5q$. 

From here on, we analyze the different clusters of $r'$ according to properties of their target tuples $v$. There are four types:
\begin{itemize}
    \item[(a)] No pair of values in $v$ co-occurs in any triple of $S$.
    \item[(b)] Exactly one pair in $v$ co-occurs in a triple of $S$.
    \item[(c)] Exactly two pairs in $v$ co-occur in triples of $S$.
    \item[(d)] $v$ is an $S$-tuple.
\end{itemize}
Due to the TF property, these types are pairwise disjoint, and they cover all possibilities of clusters of $r'$. We show a bound on the savings of each cluster, according to its type. A numerical analysis then shows that, to reach the savings of $5q$ (or more), there can be no clusters of type (a)-(c), and there are precisely $q$ clusters of type (d). As aforesaid, the corresponding $q$ target tuples form a perfect matching, as required.

\subparagraph*{Beyond three attributes.}
To extend this construction to relations of arity $c$, we select three arbitrary attributes to carry the R3DM gadgets and their matching logic. For each of the remaining $c-3$ attributes, we assign a unique fresh constant to every tuple, and guard this constant with a dedicated enforcing group. This forces every consistent update within the budget to modify all values in these additional attributes, at a fixed, uniform cost. Since this extra cost is independent of the other choices, it merely shifts the budget by a constant, leaving the equivalence with the existence of a perfect matching unaffected.

\section{Case Analysis for Small SCCs}\label{sec:small-sccs}

Theorem~\ref{thm:kc-hard} settles the FD sets that are strongly connected and have
three or more attributes. In this section, we turn to the FD sets in which every SCC of
$\graph(\depset)$ consists of one or two attributes. An SCC of a single attribute
induces no nontrivial FDs at all, and an SCC of two attributes $A$ and $B$ induces the
matching constraint $A\lra B$, which is one of the two tractable cases of
Theorem~\ref{thm:main_dichotomy}. Hence, unlike the situation of
Section~\ref{sec:large-scc}, intractability can no longer originate in a single SCC; it
can arise only from the way in which several SCCs are \e{combined}, that is, from the
FDs that hold between attributes of distinct SCCs. The FD sets that neither section
covers, namely those that have an SCC of three or more attributes without being
strongly connected themselves, are handled in Section~\ref{sec:scc-reduction}, where we
show that an FD set inherits the hardness of a convex subset and, in particular, that
of any of its SCCs.

We establish hardness for the following combinations of small SCCs:
\begin{itemize}
  \item $A\ra B\ra C$: three single-attribute SCCs forming a directed path;
  \item $A\ra B\la C$: two single-attribute SCCs that determine a third;
  \item $A\lra B\ra C$: a matching constraint that determines a single attribute;
  \item $A\ra B\lra C$: a single attribute that determines a matching constraint;
  \item $A\lra B\ra C\lra D$: a matching constraint that determines another matching
        constraint.
\end{itemize}
The
first combination, the directed path, is already known to be intractable
\cite{DBLP:conf/icdt/KolahiL09}, as discussed in
Section~\ref{subsec:state-of-affairs}, so we prove hardness for the remaining four. We
give the full proof for $A\ra B\la C$; for the other three, we give only a proof
sketch here, and the complete construction and formal analysis are provided in the
appendix.

\subsection{Hardness of $A\ra B\la C$}\label{sec:flaw-acbc}

We begin with $A\ra B\la C$, that is, the FD set $\set{A\ra B,\,C\ra B}$, in which two
independent attributes determine a third. As mentioned in
Section~\ref{subsec:state-of-affairs}, Miao et
al.~\cite{DBLP:journals/vldb/MiaoZLWC23} have already claimed that this FD set is
intractable. They give two arguments, and we begin by explaining why each of them is
flawed.

The first argument is based on the \e{conflict hypergraph} of a relation $r$: the
vertices are the cells of $r$, and the hyperedges are the sets of cells that jointly
witness a violation of an FD of $\depset$. The authors construct a reduction \e{from} the optimal U-repair problem \e{to} minimum
vertex cover over conflict hypergraphs, and intractability is then deduced from the
NP-hardness of the latter. This deduction reverses the direction of the reduction. Observe also
that the hardness argument is stated for every $\depset$ that contains two FDs $X\ra C$ and
$Y\ra C$ with $X\neq Y$, and in this generality it cannot hold (unless $\mbox{P}=\mbox{NP}$): the FD set
$\set{A\ra C,\,AB\ra C}$ satisfies this condition, yet its left-hand sides form a
chain, since $A\subseteq AB$, and so an optimal U-repair for it can be computed in
polynomial time~\cite{DBLP:journals/tods/LivshitsKR20}.

The second argument is a direct reduction from \mathsc{Max NM-E3SAT}, the problem of
satisfying as many clauses as possible of a $3$CNF formula $\phi$ in which every
clause consists of either three positive literals or three negative literals. Given
such a formula with $m$ clauses, the constructed relation over $\set{A,B,C}$ contains,
for every clause $c$ and every variable $x$ occurring in $c$, the tuple $(c,1,x)$ if
$x$ occurs positively in $c$, and the tuple $(c,0,x)$ if it occurs negatively. The
analysis concludes that the minimum cost of a consistent update of this relation is
precisely $3m-s$, where $s$ is the maximum number of clauses of $\phi$ that a single
assignment satisfies; as $s\le m$, this cost is at least $2m$. The conclusion is
false. Setting the $B$-value of \e{every} tuple to $1$ makes the $B$-column constant
and, therefore, already produces a consistent relation; its cost is the number of
tuples that carry $0$, that is, three times the number of all-negative clauses.
Symmetrically, setting every $B$-value to $0$ costs three times the number of
all-positive clauses. One of the two is of cost at most $3m/2$, which is strictly
below the claimed minimum.

We now establish the hardness of $A\ra B\la C$.

\begin{proposition}\label{prop:in-star-hard}
Computing an optimal U-repair for $A\ra B\la C$ is NP-hard.
\end{proposition}

\begin{proof}
We reduce from the multiway-cut problem with three terminals on bipartite graphs. Let
$G=(U,W,E)$ be a bipartite graph with sides $U$ and $W$, and let
$S=\set{s_1,s_2,s_3}\subseteq U\cup W$ be three distinct \e{terminals}. A \e{multiway
cut} of $(G,S)$ is a set $E'\subseteq E$ such that no two terminals are connected in
the graph obtained from $G$ by deleting the edges of $E'$. Deciding whether $(G,S)$
has a multiway cut of size at most $k$ is NP-complete already for three
terminals~\cite{doi:10.1137/S0097539792225297}. It remains NP-complete on bipartite graphs: subdividing every edge once (i.e.,~replacing each edge $u-v$ by a path $u-x_{uv}-v$ through a new vertex $x_{uv}$) makes the graph bipartite, changes neither the terminals nor the minimum size of a multiway cut, and preserves NP-completeness~\cite{johnson_et_al:LIPIcs.SWAT.2024.29}.

\partitle{Construction}
Let $(G,S,k)$ be such an instance, with $G$ bipartite, and denote $m=|E|$ and
$n=|U|+|W|$. We may assume that $k\le m$, since otherwise $E$ itself is a multiway cut
of size at most $k$. We construct a relation $r$, as illustrated in
Figure~\ref{fig:in-star}, over $\set{A,B,C}$ and a budget $K_r$, writing
a tuple as a triple $(a,b,c)$ that lists its values for $A$, $B$ and $C$.
The
attribute $B$ holds one of three \e{colors} $1$, $2$ and $3$, one for each terminal.
We set $H\defeq 2k+1$. The relation $r$ consists of the following tuples.
\begin{itemize}
  \item \e{Anchors.} If $v\in U\cup W$ and  $b\in\set{1,2,3}$, then we call the tuple $(v,b,v)$ an \e{anchor} of $v$. We add to $r$:
        \begin{itemize}
        \item $H$ anchors for every vertex $v$ and color $i\in\set{1,2,3}$
        (hence, $3H$ tuples in total for $v$). 
        \item  Additional $H$ anchors
        $(s_i,i,s_i)$ for every terminal $s_i\in\set{s_1,s_2,s_3}$.
        \end{itemize}
  \item \e{Edge gadgets.} For every edge $\set{u,w}\in E$ with $u\in U$ and $w\in W$,
        the relation contains the three tuples $(u,1,w)$, $(u,2,w)$ and $(u,3,w)$.
\end{itemize}
The FD set $\depset$ requires that, whenever two tuples agree on $A$ or on $C$, they should also agree on the color in $B$. Hence, the multiplicity of colors is what causes the violations in $r$.

Finally, we set $K_r\defeq 2nH+2m+2k$. Both $r$ and $K_r$ are clearly computable in
polynomial time. We need to show that $(G,S)$ has a multiway cut of size at most $k$ if and
only if $r$ has a consistent update $r'$ with $\distu(r',r)\le K_r$.

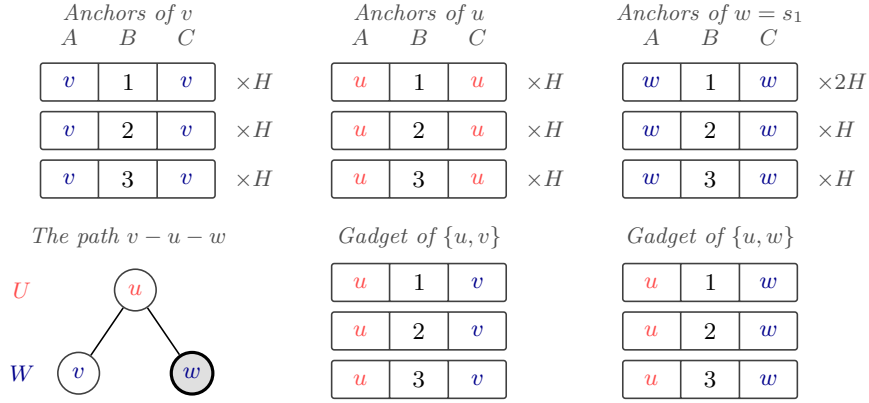
\begin{figure}[t]
\centering
\begin{tikzpicture}[
    cell/.style={rectangle split, rectangle split parts=3,
                 rectangle split horizontal, draw=black!75, semithick,
                 rounded corners=1pt,
                 minimum height=5mm, text width=6.5mm, align=center,
                 font=\small, inner sep=1.5pt, fill=white},
    colname/.style={font=\footnotesize\itshape, text=black!70},
    blockname/.style={font=\footnotesize\itshape, text=black!70},
    ann/.style={font=\footnotesize, anchor=west, text=black!70},
    vrtx/.style={circle, draw=black!75, semithick, minimum size=5.5mm,
                 inner sep=0pt, font=\small, fill=white},
    trmnl/.style={vrtx, draw=black, very thick, fill=black!12}
  ]
  \def\fl#1{\textcolor{blue!55!black}{#1}}
  \def\el#1{\textcolor{red!70}{#1}}

  \node[cell] (av1) at (0,0)
    {$\fl{v}$ \nodepart{two} $1$ \nodepart{three} $\fl{v}$};
  \node[cell, below=1.2mm of av1] (av2)
    {$\fl{v}$ \nodepart{two} $2$ \nodepart{three} $\fl{v}$};
  \node[cell, below=1.2mm of av2] (av3)
    {$\fl{v}$ \nodepart{two} $3$ \nodepart{three} $\fl{v}$};
  \node[ann, right=1.2mm of av1] {$\times H$};
  \node[ann, right=1.2mm of av2] {$\times H$};
  \node[ann, right=1.2mm of av3] {$\times H$};

  \node[cell] (au1) at (3.85,0)
    {$\el{u}$ \nodepart{two} $1$ \nodepart{three} $\el{u}$};
  \node[cell, below=1.2mm of au1] (au2)
    {$\el{u}$ \nodepart{two} $2$ \nodepart{three} $\el{u}$};
  \node[cell, below=1.2mm of au2] (au3)
    {$\el{u}$ \nodepart{two} $3$ \nodepart{three} $\el{u}$};
  \node[ann, right=1.2mm of au1] {$\times H$};
  \node[ann, right=1.2mm of au2] {$\times H$};
  \node[ann, right=1.2mm of au3] {$\times H$};

  \node[cell] (aw1) at (7.7,0)
    {$\fl{w}$ \nodepart{two} $1$ \nodepart{three} $\fl{w}$};
  \node[cell, below=1.2mm of aw1] (aw2)
    {$\fl{w}$ \nodepart{two} $2$ \nodepart{three} $\fl{w}$};
  \node[cell, below=1.2mm of aw2] (aw3)
    {$\fl{w}$ \nodepart{two} $3$ \nodepart{three} $\fl{w}$};
  \node[ann, right=1.2mm of aw1] {$\times 2H$};
  \node[ann, right=1.2mm of aw2] {$\times H$};
  \node[ann, right=1.2mm of aw3] {$\times H$};

  \node[colname, above=1mm of av1.one north]   {$A$};
  \node[colname, above=1mm of av1.two north]   {$B$};
  \node[colname, above=1mm of av1.three north] {$C$};
  \node[colname, above=1mm of au1.one north]   {$A$};
  \node[colname, above=1mm of au1.two north]   {$B$};
  \node[colname, above=1mm of au1.three north] {$C$};
  \node[colname, above=1mm of aw1.one north]   {$A$};
  \node[colname, above=1mm of aw1.two north]   {$B$};
  \node[colname, above=1mm of aw1.three north] {$C$};

  \node[blockname] at (0,0.9)    {Anchors of $v$};
  \node[blockname] at (3.85,0.9) {Anchors of $u$};
  \node[blockname] at (7.7,0.9)  {Anchors of $w=s_1$};

  \node[cell] (e1) at (3.85,-2.65)
    {$\el{u}$ \nodepart{two} $1$ \nodepart{three} $\fl{v}$};
  \node[cell, below=1.2mm of e1] (e2)
    {$\el{u}$ \nodepart{two} $2$ \nodepart{three} $\fl{v}$};
  \node[cell, below=1.2mm of e2] (e3)
    {$\el{u}$ \nodepart{two} $3$ \nodepart{three} $\fl{v}$};

  \node[cell] (f1) at (7.7,-2.65)
    {$\el{u}$ \nodepart{two} $1$ \nodepart{three} $\fl{w}$};
  \node[cell, below=1.2mm of f1] (f2)
    {$\el{u}$ \nodepart{two} $2$ \nodepart{three} $\fl{w}$};
  \node[cell, below=1.2mm of f2] (f3)
    {$\el{u}$ \nodepart{two} $3$ \nodepart{three} $\fl{w}$};

  \node[blockname] at (3.85,-2.05) {Gadget of $\set{u,v}$};
  \node[blockname] at (7.7,-2.05)  {Gadget of $\set{u,w}$};

  \node[vrtx]  (gu) at (0.1,-2.75)   {$\el{u}$};
  \node[vrtx]  (gv) at (-0.65,-3.85) {$\fl{v}$};
  \node[trmnl] (gw) at (0.85,-3.85)  {$\fl{w}$};
  \draw[semithick] (gv) -- (gu) -- (gw);
  \node[colname] at (-1.4,-2.75) {$\el{U}$};
  \node[colname] at (-1.4,-3.85) {$\fl{W}$};
  \node[blockname] at (0,-2.05) {The path $v-u-w$};
\end{tikzpicture}
\caption{The tuples that the reduction constructs for a path $v-u-w$ of $G$, where
$u\in U$, $v,w\in W$, and $w$ is the terminal $s_1$.}
\label{fig:in-star}
\end{figure}

\partitle{From a multiway cut to a consistent update}
Let $E'$ be a multiway cut of $(G,S)$ with $|E'|\le k$. Every connected component of
the graph $G'=(U\cup W,E\setminus E')$ contains at most one terminal, so we obtain a coloring $\lambda:U\cup W\ra\set{1,2,3}$ by setting $\lambda(v)=i$ if the component of
$v$ contains $s_i$, and $\lambda(v)=1$ if it contains no terminal. In particular,
$\lambda(s_i)=i$, and $\lambda(u)=\lambda(w)$ for every remaining edge $\set{u,w}\in E\setminus E'$.

We construct $r'$ as follows. For every vertex $v$, we change to $\lambda(v)$ the color
of every anchor of $v$ whose color differs from $\lambda(v)$. A vertex $v\notin S$ has
$2H$ such anchors that we update, and so does a terminal $s_i$, since $\lambda(s_i)=i$ and there are $2H$ anchors
of $s_i$ of the other two colors. The anchors therefore cost $2nH$
in total. Next, consider the gadget of an edge $\set{u,w}$.
\begin{itemize}
  \item If $\lambda(u)=\lambda(w)$, then we change to $\lambda(u)$ the color of the two
        tuples of the gadget whose color differs from $\lambda(u)$; the gadget becomes
        three copies of $(u,\lambda(u),w)$, at a cost of $2$.
  \item If $\lambda(u)\neq\lambda(w)$, then
  the edge $\set{u,v}$ is necessarily one of the removed tuples in $E'$. We replace the $C$-value of the tuple of
        color $\lambda(u)$ with a fresh value, we replace the $A$-value of the tuple of
        color $\lambda(w)$ with a fresh value, and we replace both the $A$-value and
        the $C$-value of the third tuple with fresh values, where every fresh value is
        used in a single cell of $r'$. The cost is $1+1+2=4$.
\end{itemize}
Thus, every edge incurs a cost of 2, and each edge in $E'$ may incur 2 more. 
Since $|E'|\le k$, we conclude that the total cost for updates in the edge gadgets is at most $2m+2k$. We conclude that $\distu(r',r)\le 2nH+2m+2k=K_r$, as required.

It remains to show that $r'\models\set{A\ra B,\,C\ra B}$. Let $a$ be a value that
occurs in the $A$-column of $r'$. If $a=v$ for a vertex $v\in U\cup W$, then the tuples with
$A=v$ are the anchors of $v$ together with the tuples of gadgets of edges incident to
$v$ that retained their $A$-value, and all of them have the color $\lambda(v)$. Every other value of the $A$-column is fresh
and occurs in a single tuple. Hence, $r'\models A\ra B$. The argument for
$C\ra B$ is symmetric.

\partitle{From a consistent update to a multiway cut}
Let $r'$ be a consistent update of $r$ with $\distu(r',r)\le K_r$. We say that a tuple of
$r'$ is \e{attached} to a vertex $u\in U$ if its $A$-value is $u$, and that it is
attached to a vertex $w\in W$ if its $C$-value is $w$. Since $r'\models A\ra B$ and
$r'\models C\ra B$, all the tuples that are attached to a vertex $v$ agree on their
$B$-value; we denote this value by $\lambda(v)$. If no tuple is attached to a node $v$, then we
let $\lambda(v)$ be a value that is not a color, say $\diamond$.

We first bound the cost of the anchors. If an anchor of $v$ is unchanged in $r'$, then
it is attached to $v$, and its color is $\lambda(v)$. Consequently, the anchors
of $v$ contribute to $\distu(r',r)$ at least the number of anchors of $v$ whose color
differs from $\lambda(v)$. This number is:
\begin{itemize}
\item $2H$ if $v\notin S$ and $\lambda(v)$ is a
color, or $v\in S$ is the terminal $s_i$ and $\lambda(v)=i$.
\item $3H$ or more in every other case, that is, $\lambda(v)=\diamond$ or $v=s_i$ and $\lambda(v)\neq i$.
\end{itemize}
Summing over the $n$ vertices, the anchors cost at least
$2nH$, and at least $2nH+H$ if some vertex $v$ has a $\lambda(v)$ that is not a color
or some terminal $s_i$ has $\lambda(s_i)\neq i$.

We now bound the cost of the edge gadgets. The three tuples of the gadget of $\set{u,w}$ agree on $A$ and on $C$ but have distinct colors, so at most one of them is
unchanged in $r'$, and the gadget costs at least $2$. Therefore, the edge gadgets contribute at least $2m$ to the cost.

Hence, the anchors and the gadgets cost at least $2nH+2m$ altogether, which leaves a slack of $K_r-(2nH+2m)=2k$.
Since $H=2k+1$, the additional cost of $H$ identified above exceeds this slack, and it
follows that $\lambda$ maps every vertex to a color and that $\lambda(s_i)=i$ for
$i=1,2,3$. In particular, every vertex has at least one attached tuple in $r'$, since
otherwise its $\lambda$-value would not be a color.

Let $E'$ be the set of the edges $\set{u,w}\in E$ with $\lambda(u)\neq\lambda(w)$. We
claim that the gadget of such an edge costs at least $4$. No tuple of $r'$ has both
$A=u$ and $C=w$, since its color would then be both $\lambda(u)$ and $\lambda(w)$;
therefore, each of the three tuples of the gadget changes its $A$-value or its
$C$-value. Moreover, the tuple $t$ whose color is neither $\lambda(u)$ nor $\lambda(w)$
changes at least two cells: if $t$ retains $A=u$, then it changes its color to
$\lambda(u)$ and, in addition, its $C$-value; if $t$ retains $C=w$, then it changes its
color to $\lambda(w)$ and, in addition, its $A$-value; and if $t$ retains neither, then
it changes both its $A$-value and its $C$-value. 

We conclude that $\distu(r',r)\ge 2nH+2m+2|E'|$, and $|E'|\le k$ follows from
$\distu(r',r)\le K_r$. Finally, $E'$ is a multiway cut of $(G,S)$: the endpoints of every
edge of $E\setminus E'$ have the same color, hence every connected component of
$(U\cup W,E\setminus E')$ is monochromatic, whereas the terminals have pairwise
distinct colors (i.e., $\lambda(s_i)=i$).
\end{proof}

\subsection{Remaining Proofs}

We now turn to the three remaining combinations, describing for each of them the
construction and the idea behind its analysis.

\begin{proposition}\label{prop:match-to-att-hard}
Computing an optimal U-repair for $A\lra B\ra C$ is NP-hard.
\end{proposition}

\begin{proofsketch}
We reduce from \mathsc{Vertex Cover}. Let $G=(V,E)$ be a graph and let $k\le|V|$ be a
bound on the size of a vertex cover. We set $L\defeq|V|+1$ and $M\defeq L|E|+|V|+1$,
and we construct a relation $r$ over $\set{A,B,C}$ that consists of:
\begin{itemize}
\item The \e{anchors} of
every vertex $v\in V$, consisting of $M+1$ copies of the tuple $(v,v,0)$ and $M$ copies of $(v,v,1)$;
\item An \e{edge gadget} of $L$ copies of $(u,v,1)$ for every edge
$\set{u,v}\in E$. (We select an arbitrary ordering $(u,v)$ between the nodes.)
\end{itemize}
The budget is $K_r\defeq M|V|+L|E|+k$.

Given a vertex cover $S$ of size at most $k$, we construct a repair by changing every anchor $(v,v,0)$ to $(v,v,1)$ if $v\in S$ (cost $M|S|+|S|$), changing every anchor $(v,v,1)$ to $(v,v,0)$ if $v\notin S$
(cost $M(|V|-|S|)$), and changing every
edge gadget $(u,v,1)$ to $(u,u,1)$ or $(v,v,1)$, depending on whether $u$ or $v$ is in $S$ (cost $L|E|$). All in all, the cost is $M|V|+L|E|+|S|\leq K_r$.

Now suppose that we are given a repair of cost at most $K_r$.
Since $A\lra B$, the pairs of $A$-values and $B$-values that occur in a consistent relation
form a partial bijection. The anchors of $v$ all carry the pair $(v,v)$ and split into
a group of $M+1$ tuples with $C=0$ and a group of $M$ tuples with $C=1$, violating $B\ra C$.  Hence, we must modify at least one cell in every tuple of one of them, at a cost of at least $M$ per vertex. The budget is $M|V|+L|E|+k$, so this leaves a slack of $L|E|+k$. Pulling the two groups apart is expensive: a group can be separated only by
changing \e{both} its $A$-value and its $B$-value, since changing one of the two alone
collides with the group that stays, and this costs at least $2M$. As $M$ exceeds the
entire slack, no repair can afford it, and every vertex must instead unify its
anchors on a single value $c_v\in\set{0,1}$ of $C$, at a cost of $M$ if $c_v=0$ and of
$M+1$ if $c_v=1$. The anchors therefore cost precisely $|V|M+|S|$, where
$S\defeq\setof{v\in V}{c_v=1}$ is the intended vertex cover.

The edge gadgets test whether $S$ covers $E$. A tuple $(u,v,1)$ cannot retain both
$A=u$ and $B=v$, since the anchors bind $u$ to the pair $(u,u)$ and $v$ to the pair
$(v,v)$; it must join the cluster of $u$ or the cluster of $v$. Joining that of $u$
amounts to changing its $B$-value, and it changes nothing else exactly when $c_u=1$,
and symmetrically for $v$. Hence, an edge costs $L$ if at least one of its endpoints
belongs to $S$, and at least $2L$ otherwise. Since $L>k$, the budget admits no
uncovered edge, and it leaves $|S|\le k$.
\end{proofsketch}

\begin{proposition}\label{prop:att-to-match-hard}
Computing an optimal U-repair for $A\ra B\lra C$ is NP-hard.
\end{proposition}

\begin{proofsketch}
We reduce from \mathsc{3SAT}, assuming that every clause consists of three literals
over three distinct variables. Let $\phi$ be a formula with variable set $V$ and $m$
clauses, and set $L\defeq 4m+1$. With every variable $v$ we associate two fresh values
$T_v$ and $F_v$, one for each truth value. The constructed relation over $\set{A,B,C}$
consists of a \e{variable gadget} for every variable $v$, namely $L$ copies of
$(v,v,T_v)$ and $L$ copies of $(v,v,F_v)$, and a \e{clause gadget} for every clause
$c=(\ell_1\vee\ell_2\vee\ell_3)$ over the variables $v_1,v_2,v_3$, namely the three
tuples: 
\[
  (c,v_1,w_3)\qquad (c,v_2,w_1) \qquad (c,v_3,w_2)
\]
Here, $w_i$ is the value of $v_i$ that satisfies $\ell_i$, that is, $T_{v_i}$ if
$\ell_i$ is positive and $F_{v_i}$ if $\ell_i$ is negative. Note that the indices in a clause gadget
are shifted: $v_1$, $v_2$ and $v_3$ coincide with $w_3$, $w_1$ and $w_2$, respectively. For example, if $c$ is $x\lor y\lor \neg z$, then the three tuples are:
\[
  (c,x,F_z)\qquad (c,y,T_x) \qquad (c,z,T_y)
\]
The budget is $K\defeq L|V|+4m$. 
The proof shows how a satisfying assignment can yield a repair that costs $L$ on each variable gadget (changing the truth values according to the assignment) and $4$ for each clause gadget. For the other direction, we show that with a budget of $K$, the variable gadgets must define an assignment $\tau$ as in the first direction, and each clause must spend precisely $4$ updates, implying that it retains at least one of its satisfying $w_i$'s, which is delivered by the assignment $\tau$.
\end{proofsketch}

\begin{proposition}\label{prop:match-to-match-hard}
Computing an optimal U-repair for $A\lra B\ra C\lra D$ is NP-hard.
\end{proposition}
\begin{proofsketch}
We reduce from $A\lra B\ra C$, which is the FD set of \Cref{prop:match-to-att-hard} and
consists of the first three attributes of $A\lra B\ra C\lra D$. Let $r$ be a relation
over $\set{A,B,C}$ with $n$ tuples, let $K$ be a budget, and set $N\defeq K+n+1$. We
construct a relation $r'$ over $\set{A,B,C,D}$ that consists of:
\begin{itemize}
\item A \e{main tuple} for every tuple of $r$, which copies its $A$-, $B$-, and $C$-values and uses a
fresh value in the new attribute $D$, called the \e{private} value of the tuple;
\item An \e{enforcing group} of $N$ identical tuples for every fresh
value $d$, using the value $d$ in $D$ and, in each of $A$, $B$ and $C$, a value that
occurs nowhere else.
\end{itemize}
The budget is $K+n$.

Given a consistent update $s$ of $r$ of cost at most $K$, we construct a consistent
update of $r'$ of cost at most $K+n$: we leave the enforcing groups unchanged, we
modify the $A$-, $B$- and $C$-cells of the main tuples as $s$ does, at a cost of at most
$K$, and we give the main tuples that share a $C$-value in $s$ a common fresh
$D$-value, which satisfies $C\lra D$ at the price of one cell per tuple of $r$. Since
these $D$-values are fresh, no main tuple shares a value with an enforcing group in any
attribute.

Now suppose that we are given a consistent update of $r'$ of cost at most $K+n$. As $N$
exceeds the budget, every enforcing group retains an unchanged copy. We may assume that
no main tuple agrees with an enforcing tuple on $C$ (equivalently, on $D$). Indeed,
consider a group of main tuples that have been updated to the same $(a,b,c,d)$, where
$(c,d)$ is the pair of an enforcing tuple. All of them have changed their $C$-value,
since $c$ occurs nowhere in $r$, and at most one of them retains its original
$D$-value. We take the original $C$-value $c'$ of one of them, and we pair it with the
$D$-value already used for $c'$ in the update, if $c'$ already occurs, or with a fresh
$D$-value otherwise. We replace all of their pairs $(c,d)$ by this pair $(c',d')$. If
they also carry the $A$- and $B$-values of the enforcing tuple, then they have all been updated, and we replace them by a
fresh pair; this does not increase the cost since these cells were already changed.
One $C$-cell is then restored and at most one extra $D$-cell is changed, so the cost
does not increase. Repeating for every such group, we get a repair of $r'$ where no main tuple agrees with an enforcing tuple on $D$.

It follows that no main tuple retains its private value, since it would then agree with
its own enforcing tuple on $D$. The update therefore modifies all the $D$-cells of the
main tuples and at most $K$ cells on $A$, $B$ and $C$, and their restriction to $A$,
$B$ and $C$ is an update of $r$ of cost at most $K$, which is consistent since
$A\lra B\ra C$ is implied by $A\lra B\ra C\lra D$.
\end{proofsketch}

\section{Generalization via Hard Convex Sets}\label{sec:scc-reduction}

We now generalize the hardness results of the specific cases from the previous sections to the full hardness side of \Cref{thm:main_dichotomy}, thus completing the proof of the theorem. For the rest of this section, we fix a set $\depset$ of unary FDs. 

For a set $S$ of attributes, we denote by $\depset_S$ the FD set that $S$ \e{induces}, that is, the set of all FDs in $\depset^+$ that use only attributes in $S$. 
The generalization goes as follows. 
We first define the notion of a \e{convex} set $S$ of attributes with respect to $\depset$.  Second, we show that for each convex set $S$ of $\depset$, there is a reduction from computing an optimal U-repair for $\depset_{S}$ to computing an optimal U-repair for $\depset$. Third, we show that if $\depset$ is in the hardness side of \Cref{thm:main_dichotomy}, then at least one of its convex sets $S$ induces one of the hard special cases $\depset_{S}$ of the previous sections. We begin with the definition of convexity.

\begin{definition}[Convexity]
A subset $S$ of $\attr(\depset)$ is said to be \e{convex} (\e{w.r.t.~$\depset$}) if for every three attributes $A$, $B$ and $C$ with $A\ra B\ra C\in\depset^+$, if $A\in S$ and $C\in S$ then $B\in S$. 
\end{definition}
As an example, consider
$\depset=\set{A\ra B, B\ra C, A\ra D, D\ra C}$. Then $S=\set{A,B,C}$ is \e{not} convex, since we have $A\ra D\ra C\in\depset^+$ and both $A$ and $C$ are in $S$, but $D$ is not. On the other hand, the sets $\set{A}$, $\set{A,B}$ and $\set{A,B,D}$ are all convex. Also, convexity is closely related to strong connectivity: 
\begin{observation}\label{obs:scc-convex}
The following hold. 
\begin{enumerate}
\item Every SCC is convex.
\item Every convex set is the union of SCCs (but a union of SCCs is not necessarily convex).
\end{enumerate}
\end{observation}

We prove the following:

\begin{lemma}[Reduction from a convex set]\label{lem:convex-reduction}
Let $\depset$ be a set of unary FDs, and let $S\subseteq\attr(\depset)$ be a convex set of attributes. There is a polynomial-time cost-preserving reduction from 
computing an optimal U-repair for $\depset_S$ to computing an optimal U-repair for $\depset$.
\end{lemma}
\begin{proof}
\def\theconst{\diamond}
We fix a constant $\theconst$ in $\vals$ that we use throughout the proof. An attribute $B\in R\setminus S$ is said to be \e{determined} by $S$ if $A\ra B\in\depset^+$ for some $A\in S$. 
Given a relation $s$ over $S$, we construct a relation $r$ over $R=\att(\depset)$ with $\tids(s)=\tids(r)$ by setting, for every $i\in\tids(r)$:
\begin{itemize}
  \item $r[i][A]=s[i][A]$ for each attribute $A\in S$;
  \item $r[i][B]=\theconst$ for each attribute $B$ determined by $S$;
  \item $r[i][C]$ is a fresh new value for each attribute that is neither in $S$ nor determined by $S$.
\end{itemize}
The construction is \e{correct} in the sense that to repair $r$, it is necessary and sufficient to repair the projection of $r$ to $S$, namely $s$. Hence, a U-repair for $r$ yields a U-repair (of the same or smaller cost) for $s$, and a U-repair for $s$ yields a U-repair (of the same cost) for $r$. 

We prove correctness as follows. 
Let $i$ and $j$ be identifiers in $\tids(r)$. We need to show that every FD $A\ra B$ violated by $\set{r[i], r[j]}$ is such that both $A$ and $B$ are in $S$ (hence, $\depset_S\models A\ra B$). Thus, it suffices to show that if either $A\notin S$ or $B\notin S$, then the FD $A\ra B$ is satisfied, regardless of the content of $s$. (In particular, it continues to be satisfied even if $s$ is updated.) We consider several cases.
\begin{enumerate}
    \item If $B\in S$, then $A\notin S$ and $A$ cannot be determined by $S$, or otherwise $A$ would need to be in $S$ since we assume that $S$ is convex. Therefore, the rows disagree on $A$ because $r[i][A]$ and $r[j][A]$ are fresh new values (by construction).
    \item If $B$ is determined by $S$, then the rows agree on $B$, since $r[i][B]=r[j][B]=\theconst$.
    \item If $B$ is neither in $S$ nor determined by $S$, then $A$ is also neither in $S$ nor determined by $S$, or otherwise $A\ra B$ would imply that $B$ is determined by $S$. Hence, the rows disagree on $A$ as in the first case.
\end{enumerate}
In each of the three cases, the FD $A\ra B$ is satisfied, as claimed.
\end{proof}

From \Cref{obs:scc-convex} and \Cref{lem:convex-reduction}, we immediately conclude that it is NP-hard to compute an optimal U-repair for $\depset$ if it is NP-hard for at least one SCC of $\depset$. The other direction, however, is false (assuming $\mbox{P}\neq\mbox{NP}$), as evident from the hardness of $A\ra B\ra C$, where each SCC consists of a single attribute. 


Finally, we show that the hard cases of the previous sections are exhaustive, in the sense
that every FD set on the hardness side of \Cref{thm:main_dichotomy} has a convex set of
attributes that induces one of them. This, together with \Cref{lem:convex-reduction}, completes the proof of \Cref{thm:main_dichotomy}.

\begin{lemma}[Existence of a hard convex set]\label{lem:hard-convex-set}
Let $\depset$ be a set of nontrivial unary FDs that does not fall in the tractable side
of \Cref{thm:main_dichotomy}. Then $\depset$ has a convex set $S$ of attributes such
that, up to a renaming of the attributes, $\depset_S$ is one of the FD
sets $A\ra B\ra C$, $A\ra B\la C$, $A\lra B\ra C$, $A\ra B\lra C$,
$A\lra B\ra C\lra D$, and $\depset_{K_c}$ for some $c\ge3$.
\end{lemma}
\begin{proof}
Without loss of generality, we assume that $\depset$ is weakly connected. Otherwise, we apply the proof to a weakly connected component of $\depset$ where the tractability condition is violated. Clearly, an attribute set that is convex w.r.t.~a weakly connected component of $\depset$ is also convex w.r.t.~$\depset$.

We use the \e{SCC DAG} of $\depset$: its nodes are the SCCs of $\depset$, and it has an edge from a node $P$ to a node $Q$ whenever $P\neq Q$ and the attributes of $P$ determine those of $Q$. Note that the SCC DAG is connected, since we assume that $\depset$ is weakly connected. We write $P\prec Q$ to state that the SCC DAG has an edge from $P$ to $Q$. Note that $\prec$ is transitive, asymmetric, and irreflexive, hence a strict partial order. 
In particular, every nonempty set of SCCs has a $\prec$-minimal member and a $\prec$-maximal member. 

We say that a node $R$ is \e{between} $P$ and $Q$ if $P\prec R\prec Q$. Note that a union of SCCs is convex if and only if no SCC outside the union is between two SCCs in the union. 
Indeed, if $\depset\models A\ra B\ra C$ and $A$ and $C$ belong to the union, then the SCC of $B$ is either that of $A$ or $C$, or between them (hence, in the union).

\def\newcase#1#2{\medskip\par\noindent \textbf{Case #1:} \e{#2}.}

Since $\depset$ violates the condition of \Cref{thm:main_dichotomy}, it has at least three attributes.
We consider three cases, according to the maximal number of attributes in an SCC of
$\depset$.

\newcase{1}{Some SCC $P$ has three or more attributes}
By   \Cref{obs:scc-convex}, the set $P$ is convex, and $\depset_P$ is $\depset_{K_c}$ for $c=|P|\ge3$.

\newcase{2}{The maximal number of attributes in an SCC is two}
Let $P$ be an SCC with two attributes, denoted $A$ and $B$. Then 
$\depset_P$ is $A\lra B$. 
%

Suppose first that $P\prec Q'$ for some SCC $Q'$. Let $Q$ be a $\prec$-minimal member of the set of nodes $Q'$ with $P\prec Q'$. The set $S\defeq P\cup Q$ is convex: 
a node between $P$ and $Q$ would belong to the set from which
we chose $Q$ and would precede $Q$, contradicting the minimality of $Q$. 
\begin{itemize}
\item If $Q$ consists of a single attribute
  $C$, then $\depset_S$ is $A\lra B\ra C$. 
\item If $Q$ consists of two attributes $C$ and
  $D$, then $\depset_S$ is $A\lra B\ra C\lra D$. 
\end{itemize}
If $Q'\prec P$ for some SCC $Q'$, then we select $Q$ as a $\prec$-maximal member of the set of the nodes $Q'$ with $Q'\prec P$. The set $S\defeq P\cup Q$ is again convex.
\begin{itemize}
\item If $Q$ consists of a single attribute
  $C$, then $\depset_S$ is $C\ra A\lra B$. 
\item If $Q$ consists of two attributes $C$ and
  $D$, then $\depset_S$ is $C\lra D \ra A\lra B$. 
\end{itemize}

\newcase{3}{Every SCC has a single attribute} In this case, the SCC DAG is the same as $\graph(\depset)$, identifying each singleton with its single member. Hence, $\graph(\depset)$ is acyclic and weakly connected. Fix a topological order over the attributes. 

We first claim that at least one attribute $C$ has two or more incoming edges. Let $A'$ be the first (leftmost) node in the topological order. If $A'$ has an outgoing edge to every other attribute, then at least one neighbor $C$ of $A'$ has an incoming edge from an attribute that is not $A'$, or else $\depset$ would be equivalent to $A'\ra X$ (satisfying the tractability condition). Otherwise, if $A'$ does not have an outgoing edge to every other attribute, then at least one neighbor $C$ of $A'$ has an incoming edge from an attribute that is not $A'$, since $\graph(\depset)$ is connected and transitively closed. This proves the claim.

Choose an arbitrary node $C$ with two or more incoming edges in $\graph(\depset)$. Let $B$ be the rightmost attribute in the topological order with an edge to $C$. Let $A$ be the rightmost attribute, different from $B$, with an edge to $C$. Let $S=\set{A,B,C}$. Hence, the FD set $\depset_S$ is either
$A\ra B\ra C$ or $A\ra C\la B$, depending on whether $A\ra B$ or not. To complete the proof, we need to show that $S$ is convex. This holds due to the choice of $B$ and $A$ as rightmost attributes: if $D\notin S$ is such that 
if $B\ra D \ra C$, then $D$ should have been chosen instead of $B$; and if $A\ra D \ra B$ or $A\ra D \ra C$, then $D$ should have been chosen instead of $A$. 
\end{proof}

With \Cref{lem:hard-convex-set}, we have now completed the proof of \Cref{thm:main_dichotomy}.

\section{Conclusions}
\label{sec:conclusions}
We have established a complete complexity classification for computing optimal U-repairs under unary FDs. Our result shows that the previously known tractable cases are exhaustive: an optimal U-repair can be computed in polynomial time when every weakly connected component of the FD set has a common lhs or is a matching constraint, and the problem is NP-hard in every other case. Thus, for unary FDs, the boundary between tractability and intractability is now fully understood.

We note that the tractable cases remain tractable even if every cell has an individual weight that determines its cost of update~\cite{DBLP:journals/tods/LivshitsKR20}. Since our results show that these are the only tractable cases even without weights, we conclude that the dichotomy trivially extends to the weighted case.

The main question left open is whether a similar classification can be obtained for arbitrary FDs. The unary case already exhibits a rich interaction between the structure of the dependencies and the complexity of repairing the data, and several of our hardness arguments rely on phenomena that extend beyond unary FDs. It would be interesting to understand whether these techniques, together with the tractability tools developed in previous work, can lead to a complete dichotomy for optimal U-repairs under general FD sets.

\bibliography{citations.bib}

\appendix

\section{Subset Repairs versus Update Repairs}

It is instructive to contrast the situation with that of \e{subset} repairs,
which are obtained by deleting a minimum number of tuples rather than by
updating cells. For subset repairs, Livshits et
al.~\cite{DBLP:journals/tods/LivshitsKR20} obtained a \e{complete} dichotomy:
for \e{every} set of FDs, computing an optimal subset repair is either solvable
in polynomial time or NP-hard (in fact, APX-complete). Their hardness argument proceeds
in two steps. First, they establish hardness \e{directly} for a small
collection of ``core'' FD sets, each by a reduction from a well-known NP-hard
problem. Second, they lift this hardness to every remaining intractable FD set
by means of \e{fact-wise reductions}: a fact-wise reduction maps each tuple of a
relation over one schema to a tuple of a relation over another schema in a way
that preserves both consistency and inconsistency. Such a
map induces a one-to-one correspondence between the subset
repairs of the source relation and those of the target, and hardness therefore
transfers immediately from the core sets to all others.

This machinery is unavailable for update repairs, which is a central reason the
update-repair picture has remained so much coarser. A fact-wise reduction
relates two relations over possibly \e{different} schemas, and in order to
encode the source dependencies it typically represents the value of a single
source attribute using \e{several} attributes of the target relation, so that
one source value determines several target values. Two obstacles then arise.

First, a single cell update in the source relation forces an entire \e{group}
of cell updates in the target relation, and the size of this group may differ
from one source attribute to another. The number of modified target cells is
therefore not a fixed multiple of the number of modified source cells, so the
reduction cannot preserve repair cost up to a fixed factor, and the tight,
cost-preserving correspondence on which the subset-repair argument relies
already fails.
Second, and more fundamentally, the correspondence breaks down in the reverse
direction. An update repair of the target relation is under no obligation to
respect these groups: it may modify the target cells that are determined by a
single source cell \e{inconsistently}, changing some of them but not others, or
changing them to values that no longer agree. Such a target repair reflects no
update of the source relation at all, since there each cell is atomic and is
either left unchanged or updated as a whole. There is thus no faithful way to
translate an optimal update repair of the target relation back into one of the
source relation of comparable cost, and the fine-grained, cell-level cost
accounting that defines an \e{optimal} update repair is not preserved. For these
reasons, hardness cannot simply be transported from a handful of core FD sets to
all the rest by means of fact-wise reductions, and different machinery is required.

\section{Missing Proofs of Section~\ref{sec:large-scc}}\label{sec:app-large-scc}

This section proves Theorem~\ref{thm:kc-hard}: for every $c\ge3$, computing an
optimal U-repair for $\scdepset{c}$ is NP-hard. The proof consists of three steps. We
first show that \mathsc{R3DM}, the restricted variant of three-dimensional matching
from which we reduce, is NP-complete (Lemma~\ref{lem:r3dm}). We then reduce
\mathsc{R3DM} to our problem for three attributes (Lemma~\ref{lem:k3-hard}). Finally,
we reduce the problem for three attributes to the problem for $c$ attributes, for every
$c>3$ (Lemma~\ref{lem:k3-to-kc}). Throughout, we specify a tuple by listing its values
for $A_1,\dots,A_c$ in this order.

\subsection{Hardness of \mathsc{R3DM}}

We first prove that the restricted matching problem that we reduce from is intractable.
Recall that an instance of \mathsc{R3DM} consists of three disjoint sets $X$, $Y$ and
$Z$ of size $q$ and a set $S\subseteq X\times Y\times Z$ of triples that satisfies the
two conditions LI and TF, and that the goal is to decide whether $S$ has a perfect
matching, that is, $q$ pairwise disjoint triples.

\begin{lemma}\label{lem:r3dm}
\mathsc{R3DM} is NP-complete.
\end{lemma}

\begin{proof}
Membership in NP is clear, so we prove hardness by a reduction from \mathsc{3DM}.
Let $(X,Y,Z,S)$ be an instance of \mathsc{3DM}, where $|X|=|Y|=|Z|=q$ and $|S|=m$. For
every triple $s=(x,y,z)\in S$ we introduce six fresh elements, namely $a^s_1$ and
$a^s_2$ that we add to $X$, $b^s_1$ and $b^s_2$ that we add to $Y$, and $c^s_1$ and
$c^s_2$ that we add to $Z$; these $6m$ elements are distinct from one another and from
the elements of $X\cup Y\cup Z$. We call them the \e{internal} elements of $s$ and the
elements of $X\cup Y\cup Z$ the \e{original} ones. We replace $s$ by the \e{gadget}
$G_s$ that consists of the five triples
\[
  \sigma^s_1=(x,b^s_1,c^s_1),\quad
  \sigma^s_2=(a^s_1,y,c^s_2),\quad
  \sigma^s_3=(a^s_2,b^s_2,z),\quad
  \sigma^s_4=(a^s_1,b^s_2,c^s_1),\quad
  \sigma^s_5=(a^s_2,b^s_1,c^s_2).
\]
The three sides of the constructed instance are of size $q+2m$ each, the constructed
set $\hat{S}$ of triples is of size $5m$, and a perfect matching consists of $q+2m$
triples. The construction is clearly in polynomial time.

 Next, we show that $\hat{S}$ satisfies LI and TF.
We use two properties of the construction. First, every triple of $\hat{S}$ contains at
most one original element, since $\sigma^s_1$, $\sigma^s_2$ and $\sigma^s_3$ contain one
each and $\sigma^s_4$ and $\sigma^s_5$ contain none. Second, the internal elements of
$s$ are introduced for $s$ alone, so they occur only in the triples of $G_s$.

For LI, consider two distinct triples of $\hat{S}$. If they belong to distinct gadgets,
then, by the second property, every element that they share is original, and, by the
first, each of them has at most one such element; hence, they share at most one element.
If they belong to the same gadget $G_s$, then a direct inspection of the five triples
above shows that they share at most one element; the inspection also shows that they
share an element precisely when one of them is $\sigma^s_1$, $\sigma^s_2$ or
$\sigma^s_3$ and the other is $\sigma^s_4$ or $\sigma^s_5$.

For TF, suppose that $\hat{S}$ contains three distinct triples $t_1$, $t_2$ and $t_3$
and three distinct elements $u$, $v$ and $w$ such that $u,v\in t_1$, $u,w\in t_2$ and
$v,w\in t_3$. Every one of the three triples contains two of the three elements, so at
most one of $u$, $v$ and $w$ is original: if two of them were, then some triple would
contain both, contradicting the first property. At least two of the elements are
therefore internal, say $u$ and $v$. By the second property, $t_1$ and $t_2$ belong to
the same gadget, since they share the internal element $u$, and so do $t_1$ and $t_3$,
since they share the internal element $v$; thus, all three belong to a single gadget
$G_s$. This is impossible: by the inspection above, no two of $\sigma^s_1$,
$\sigma^s_2$ and $\sigma^s_3$ share an element and neither do $\sigma^s_4$ and
$\sigma^s_5$, so among any three triples of $G_s$ there are two that share no element.

Now, given a perfect matching of $S$, we construct one of  $\hat{S}$.
Let $M\subseteq S$ be a perfect matching of the original instance, so $|M|=q$. We
select the triples $\sigma^s_1$, $\sigma^s_2$ and $\sigma^s_3$ for every $s\in M$, and
the triples $\sigma^s_4$ and $\sigma^s_5$ for every $s\in S\setminus M$. Each of the
two selections covers the six internal elements of $G_s$ exactly once, and the first
covers, in addition, the three original elements of $s$. Since $M$ is a perfect
matching, every original element is covered exactly once, and the number of selected
triples is $3q+2(m-q)=q+2m$, as required.

Conversely, let $\hat{M}$ be a perfect matching of $\hat{S}$. We construct a perfect matching of $S$. Let
$s\in S$. As noted above, the six internal elements of $s$ occur only in the triples of
$G_s$, so $\hat{M}$ covers each of them by a triple of $G_s$; we determine which
triples of $G_s$ it uses. If $\sigma^s_4\in\hat{M}$, then $a^s_1$, $b^s_2$ and $c^s_1$
are covered,
and the only triple of $G_s$ that covers $a^s_2$ without reusing them is $\sigma^s_5$;
the two together cover the six internal elements, and none of $s$'s original elements
is covered by $G_s$. If $\sigma^s_4\notin\hat{M}$, then $a^s_1$ forces
$\sigma^s_2\in\hat{M}$, the element $b^s_2$ forces $\sigma^s_3\in\hat{M}$, and $c^s_1$
forces $\sigma^s_1\in\hat{M}$; these three cover the six internal elements, and they
cover the original elements of $s$ as well. Hence, every gadget either covers all three
original elements of its triple or covers none of them. As every original element is
covered exactly once, the set $M$ of the triples $s$ whose gadget covers them is a
perfect matching of the original instance.
\end{proof}

\subsection{Proof of Theorem~\ref{thm:kc-hard}}

We first prove the statement for three attributes.

\begin{lemma}\label{lem:k3-hard}
Computing an optimal U-repair for $\scdepset{3}$ is NP-hard.
\end{lemma}

\begin{proof}
We reduce from \mathsc{R3DM}, which is NP-complete by Lemma~\ref{lem:r3dm}. Let
$(X,Y,Z,S)$ be an instance of R3DM, where $|X|=|Y|=|Z|=q$ and $|S|=m$, and
where we may assume that $m\ge q\ge1$, since otherwise $S$ has no perfect matching. We
refer to the members of $X\cup Y\cup Z$ as \e{$S$-elements}. We construct a relation $r$
over $A_1$, $A_2$ and $A_3$ and a budget $K$ such that $S$ has a perfect matching if
and only if $r$ has a consistent update $r'$ with $\distu(r',r)\le K$. For $c=3$, the
target tuple of a cluster is a triple, and we call it the \e{target triple} of the
cluster; recall that the target triples of distinct clusters disagree on $A_1$, $A_2$
and $A_3$.

We set:
\[
  K\;\defeq\;8m-5q
  \qquad\text{and}\qquad
  N\;\defeq\;K+1 .
\]
For every triple $s=(x,y,z)\in S$, the relation $r$ contains the four \e{gadget tuples}
\[
  t_s=(x,y,z)\qquad
  u^1_s=(c^s_1,y,z)\qquad
  u^2_s=(x,c^s_2,z)\qquad
  u^3_s=(x,y,c^s_3)
\]
where $c^s_1$, $c^s_2$ and $c^s_3$ are \e{fresh constants}. We refer to $t_s$ as an \e{$S$-tuple}, and to $u^1_s$, $u^2_s$ and $u^3_s$ its \e{auxiliary tuples}; note that
$u^j_s$ contains $c^s_j$ in $A_j$ and agrees with $t_s$ on the other two attributes. In
addition, $r$ contains, for every $j\in\set{1,2,3}$, an \e{enforcing group}: a set of
$N$ copies of a tuple that contains $c^s_j$ in $A_j$ and, in the two other attributes,
the values $d^s_{j,i}$; that is, $N$ copies of each of the following three tuples:
\[
  (c^s_1,d^s_{1,2},d^s_{1,3})\qquad
  (d^s_{2,1},c^s_2,d^s_{2,3})\qquad
  (d^s_{3,1},d^s_{3,2},c^s_3)\qquad
\]
The $3m$ fresh constants $c^s_j$ and the $6m$ values $d^s_{j,i}$ are distinct from one
another and from the elements. Figure~\ref{fig:construction} depicts the tuples
introduced for a single triple $s$, with the superscript $s$ omitted. Altogether, $r$
consists of the $4m$ gadget tuples and of the $3mN$ tuples of the enforcing groups, and
it is constructed in polynomial time.

Now, given a perfect matching of $S$, we construct a consistent update of $r$.
Let $M\subseteq S$ be a perfect matching, and recall that $|M|=q$. We leave every enforcing group unchanged, so that each of its tuples is its own target triple. 
\begin{itemize}
\item For $s=(x,y,z)\in M$, we assign the target triple $(x,y,z)$ to the four gadget tuples of $s$. The $S$-tuple $t_s$ is then left unchanged, and each auxiliary tuple $u^j_s$ changes the single cell that contains $c^s_j$, at a total \underline{cost of $3$} for each of the $q$ triples in $M$.

\item Now let $s'=(x',y',z')\in S\setminus M$. As $M$ covers every $S$-element, it contains a triple $s_{x'}=(x',y_{x'},z_{x'})$ that includes $x'$, and it contains a triple $s_{y'}=(x_{y'},y',z_{y'})$ that includes $y'$. Both tuples are distinct from $s'$, so, by LI, each of them shares with $s'$ only one element ($x'$ or $y'$), and hence 
$x_{y'}\neq x'$, $y_{x'}\neq y'$, $z_{x'}\neq z'$ and $z_{y'}\neq z'$. We assign:
\begin{itemize}
\item The target triple $s_{x'}$ to $t_{s'}$, $u^2_{s'}$ and $u^3_{s'}$;
\item The target triple $s_{y'}$ to $u^1_{s'}$.
\end{itemize}
Since no fresh constant is an $S$-element, each of $t_{s'}$, $u^2_{s'}$ and $u^3_{'s}$ keeps only its value $x'$ in $A_1$, and $u^1_{s'}$ keeps only its value $y'$ in $A_2$. Hence, each of the four tuples in the gadget changes two cells, at a total cost of $8$. Hence, we get a total \underline{cost of $8$} for each of the $m-q$ triples $s'$ in $S\setminus M$. 
\end{itemize}
In conclusion, the resulting update $r'$ satisfies
$\distu(r',r)=3q+8(m-q)=8m-5q=K$, as required.  It remains to verify that $r'$ is consistent, that is, that the target triples of distinct clusters disagree on $A_1$, $A_2$ and $A_3$. These target triples are the $q$ triples of $M$ and the $3m$ tuples of the enforcing groups. Two triples of $M$ are
disjoint and therefore disagree on every attribute; a triple of $M$ contains elements only, whereas a tuple of an enforcing group contains no elements; and two tuples of distinct enforcing groups disagree on every attribute, since the values $c^s_j$ and
$d^s_{j,i}$ are all distinct.

Next, given a consistent update of $r$, we construct a perfect matching of $S$.
Let $r'$ be a consistent update of $r$ with $\distu(r',r)\le K$. Then $r'$ includes several clusters, each consisting of copies of a target tuple.  We will show that at least $q$ of the target tuples must be $S$-tuples. From the consistency of $r'$ we can then conclude that these $q$ $S$-tuples form a perfect matching $M$, as required. We begin with two observations. 
\begin{itemize}
  \item[(O1)] \e{We can assume that every enforcing group is unchanged in $r'$.} Since $N>K$, at least one tuple from the enforcing group remains unchanged, so if any tuple in this group is changed, we can reverse its changes without causing any violation.
\item[(O2)] \e{We can further assume that every gadget tuple keeps at least one of its three cells.} Let $t\in r$ be a gadget tuple, and suppose that all three values of $t$ are updated in $r'$. If any target tuple $t'\in r'$ coincides with $t$ on at least one attribute, then we can update $t$ to $t'$ and reduce the cost by one. If no such $t'$ agrees, then we can again reverse the updates of $t'$ and reduce the cost.
\end{itemize}
The cost of $r'$ is then the number of cells changed within the gadget tuples. By (O2), we can assume that the cost $c_t$ incurred by each individual gadget tuple $t$ is at most 2. For convenience, in the remainder of this proof, we refer to $2-c_t$ as the \e{saving} of $t$. We further refer to the \e{saving of a cluster} as the total savings of all tuples in the cluster. With zero savings, we have $4m$ tuples with a total cost of $8m$. Since the total cost of $r'$ is at most $K=8m-5q$, we conclude that the total savings over all clusters are at least $5q$. 

Next, we analyze the savings of clusters by the type of their target tuple. Before that, we need another observation. 
\begin{itemize}
  \item[(O3)] \e{We can assume that whenever a target tuple is not equal to an enforcing tuple, then each value is either an $S$-element or a new value that does not occur in $r$.} This means that values of the form $d^s_{j,i}$ and $c^s_j$ cannot be used in clusters where the target tuple is not equal to an enforcing tuple. This follows directly from (O1), since enforcing tuples from $r$ remain intact in $r'$; hence, if any $d^s_{j,i}$ or $c^s_j$ is retained in a gadget tuple, then the tuple needs to agree on all attributes with the enforcing tuple that contains $d^s_{j,i}$ or $c^s_j$.
\end{itemize}

We consider five types of clusters, each defined by a property of its target cluster $v$. For each cluster, we compute a bound on its savings.




\begin{figure}[t]
\centering
\begin{tikzpicture}[
    cell/.style={rectangle split, rectangle split parts=3,
                 rectangle split horizontal, draw=black!75, semithick,
                 rounded corners=1pt,
                 minimum height=5mm, text width=7mm, align=center,
                 font=\scriptsize, inner sep=1.5pt, fill=white},
    target/.style={cell, fill=black!12, draw=black, very thick},
    rowname/.style={font=\scriptsize, anchor=east},
    save/.style={font=\scriptsize, anchor=west},
    total/.style={font=\scriptsize\bfseries, anchor=north west},
    title/.style={font=\scriptsize\itshape, text=black!70}
  ]
  \def\ag#1{\textcolor{blue!55!black}{#1}}
  \def\di#1{\textcolor{black!70}{#1}}

  \node[target] (a0) at (0,0)
    {$\ag{x}$ \nodepart{two} $\ag{y}$ \nodepart{three} $\ag{z}$};
  \node[cell, below=2.6mm of a0] (a1)
    {$\ag{x}$ \nodepart{two} $\ag{y}$ \nodepart{three} $\ag{z}$};
  \node[cell, below=1.2mm of a1] (a2)
    {$\di{c_1}$ \nodepart{two} $\ag{y}$ \nodepart{three} $\ag{z}$};
  \node[cell, below=1.2mm of a2] (a3)
    {$\ag{x}$ \nodepart{two} $\di{c_2}$ \nodepart{three} $\ag{z}$};
  \node[cell, below=1.2mm of a3] (a4)
    {$\ag{x}$ \nodepart{two} $\ag{y}$ \nodepart{three} $\di{c_3}$};
  \node[rowname, left=1mm of a0] {$v$};
  \node[rowname, left=1mm of a1] {$t_s$};
  \node[rowname, left=1mm of a2] {$u^1_s$};
  \node[rowname, left=1mm of a3] {$u^2_s$};
  \node[rowname, left=1mm of a4] {$u^3_s$};
  \node[save, right=1mm of a1] {$+2$};
  \node[save, right=1mm of a2] {$+1$};
  \node[save, right=1mm of a3] {$+1$};
  \node[save, right=1mm of a4] {$+1$};
  \node[total] at ([yshift=-1.6mm]a4.south west) {saving $5$};

  \node[target] (b0) at (4.3,0)
    {$\ag{x}$ \nodepart{two} $\ag{y}$ \nodepart{three} $\ag{z'}$};
  \node[cell, below=2.6mm of b0] (b1)
    {$\ag{x}$ \nodepart{two} $\ag{y}$ \nodepart{three} $\di{z}$};
  \node[cell, below=1.2mm of b1] (b2)
    {$\ag{x}$ \nodepart{two} $\ag{y}$ \nodepart{three} $\di{c_3}$};
  \node[cell, below=1.2mm of b2] (b3)
    {$\ag{x}$ \nodepart{two} $\di{y'}$ \nodepart{three} $\ag{z'}$};
  \node[cell, below=1.2mm of b3] (b4)
    {$\ag{x}$ \nodepart{two} $\di{c'_2}$ \nodepart{three} $\ag{z'}$};
  \node[rowname, left=1mm of b0] {$v$};
  \node[rowname, left=1mm of b1] {$t_s$};
  \node[rowname, left=1mm of b2] {$u^3_s$};
  \node[rowname, left=1mm of b3] {$t_{s'}$};
  \node[rowname, left=1mm of b4] {$u^2_{s'}$};
  \node[save, right=1mm of b1] {$+1$};
  \node[save, right=1mm of b2] {$+1$};
  \node[save, right=1mm of b3] {$+1$};
  \node[save, right=1mm of b4] {$+1$};
  \node[total] at ([yshift=-1.6mm]b4.south west) {saving $4$};

  \node[target] (c0) at (8.6,0)
    {$\ag{x}$ \nodepart{two} $\ag{y}$ \nodepart{three} $\di{w}$};
  \node[cell, below=2.6mm of c0] (c1)
    {$\ag{x}$ \nodepart{two} $\ag{y}$ \nodepart{three} $\di{z}$};
  \node[cell, below=1.2mm of c1] (c2)
    {$\ag{x}$ \nodepart{two} $\ag{y}$ \nodepart{three} $\di{c_3}$};
  \node[rowname, left=1mm of c0] {$v$};
  \node[rowname, left=1mm of c1] {$t_s$};
  \node[rowname, left=1mm of c2] {$u^3_s$};
  \node[save, right=1mm of c1] {$+1$};
  \node[save, right=1mm of c2] {$+1$};
  \node[total] at ([yshift=-1.6mm]c2.south west) {saving $2$};

  \node[title, above=3mm of a0] {$v$ is an $S$-tuple};
  \node[title, above=3mm of b0] {two pairs occur in triples};
  \node[title, above=3mm of c0] {one pair occurs in a triple};
\end{tikzpicture}
\caption{The three savings that a cluster can attain, with the target triple $v$ on top
and, below it, the gadget tuples that save through it. Cells that agree with $v$ are
shown in blue and cells that differ from it in gray, and the saving of each tuple is
written on its right. On the left, all three pairs of elements of $v$ occur in the
triple $s$; in the middle, the triples $s=(x,y,z)$ and $s'=(x,y',z')$ share only the
element $x$, and $v$ contains a pair of each; on the right, $w$ is either not an element
or an element that occurs in no triple together with $x$ or with $y$.}
\label{fig:savings}
\end{figure}
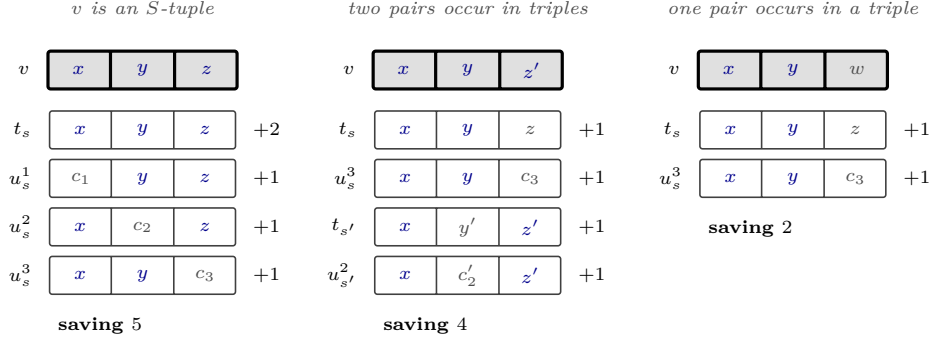

\begin{itemize}
  \item[(a)] \e{No two elements of $v$ co-occur together inside a triple of $S$.} Due to the construction of the gadget tuples, and due to (O3), every gadget tuple differs from $v$ in at least two attributes. Hence, this cluster saves  \underline{nothing}.

\item[(b)] \e{Exactly one pair of values in $v$ co-occur together in a triple of $S$.} Consider the two $S$-elements of $v$ that co-occur together in $S$. By LI, exactly one $S$-triple contains both of them. By construction, precisely two of its gadget tuples contain both---its $S$-tuple and one auxiliary tuple; if assigned to the cluster, each of them keeps two cells and saves one. 
By (O3), every other tuple in the cluster differs from $v$ in at least two attributes, so it saves nothing. So, this cluster saves \underline{at most 2} $=1+1$.
      
\item[(c)] \e{Exactly two distinct pairs of values in $v$ co-occur inside triples of $S$.} Due to TF, the third pair of values \e{does not} co-occur inside any triple of $S$. Hence, this case is similar to (b), but with twice the bound. The cluster therefore saves \underline{at most 4} $=2+2$.

\item[(d)] \e{$v$ is an $S$-tuple.} Let $s=(x,y,z)\in S$ be this tuple.
  If assigned to this cluster, the tuple $t_s$ keeps all three of its cells and saves 2, and each auxiliary tuple $u^j_s$ keeps the two cells that it inherits from $t_s$ and saves 1. Every other gadget tuple belongs to a triple $s'\neq s$, which,
        by LI, shares at most one element with $s$; hence, assigning it to the cluster saves nothing. The cluster
        saves \underline{at most 5} $=2+1+1+1$. 
\end{itemize}
Note that cases (a)-(d) are pairwise disjoint. Also, due to the TF, they cover all possible clusters. Let $n_x$ be the number of distinct target tuples of type $(x)$. To complete the proof, we will show that $n_d=q$; hence, the the target tuples of type (d) form a perfect matching.

Let $T$ be the total savings. The case analysis shows that 
$0n_a+2n_b+4n_c+5n_d \geq T$.
On the other hand, we know that $T\geq 5q$. Hence,
\begin{equation}
  2n_b+4n_c+5n_d \geq 5q\, .
\label{eq:geq-5q}
\end{equation}

Let $N$ be the number of $S$-elements that occur in the target triples. Then $3q\geq N$, which is the total number of $S$-elements.
Two distinct clusters disagree on every attribute, so no
$S$-element is contained in more than one target tuple. Each target tuple of type (b) includes at least two $S$-elements, and each target tuple of types (c) and (d) includes at least three $S$-elements. We conclude that 
$N\geq 2n_b+3n_c+3n_d$. Hence,
\begin{equation}
  3q \geq 2n_b+3n_c+3n_d.
\label{eq:3q-geq}
\end{equation}

Combining \Cref{eq:geq-5q} and \Cref{eq:3q-geq}, we get that 
\[
3q \geq 2n_b+3n_c+3n_d = 
2n_b+4n_c+5n_d - (n_c+2n_d)
\geq 5q - (n_c+2n_d)
\]
Hence, $2q \leq n_c+2n_d\leq 2n_c+2n_d$, and therefore 
$n_c+n_d \geq q$. Using \Cref{eq:3q-geq}, we conclude that
$3q \geq 2n_b+3n_c+3n_d \geq 2n_b+3q$, hence $n_b=0$. With that, from \Cref{eq:3q-geq} we  conclude that $ 3q \geq 3n_c+3n_d$, hence
$n_c+n_d \leq q$. Thus, $n_c+n_d=q$.
Finally, recalling again \Cref{eq:geq-5q}, we conclude that
\[
5q \leq 4n_c+5n_d =
4n_c+4n_d + n_d = 4q + n_d
\]
and, therefore, $n_d\geq q$ (actually,
$n_d=q$ and $n_c=0$), as claimed.
\end{proof}



It remains to pass from three attributes to any larger number, which we do by padding
every tuple with values that no consistent update within the budget can preserve.

\begin{lemma}\label{lem:k3-to-kc}
Let $c>3$. Given a relation $r$ over $A_1$, $A_2$ and $A_3$ with $n$ tuples and a
budget $K$, one can construct in polynomial time a relation $r^*$ over $A_1,\dots,A_c$
such that $r$ has a consistent update of cost at most $K$ with respect to
$\scdepset{3}$ if and only if $r^*$ has a consistent update of cost at most $K+n(c-3)$
with respect to $\scdepset{c}$.
\end{lemma}

\begin{proof}
Let $N\defeq K+n(c-3)+1$. For every tuple $t$ of $r$, the relation $r^*$ contains the
tuple $t^*$ that agrees with $t$ on $A_1$, $A_2$ and $A_3$ and contains a \e{private}
value in each of $A_4,\dots,A_c$. In addition, for every private value $p$, say the one
that $t^*$ contains in $A_j$, the relation $r^*$ contains an \e{enforcing group} of $N$
copies of the \e{enforcing tuple} $e_p$ that contains $p$ in $A_j$ and, in every other
attribute, a value that occurs nowhere else. All of these values are distinct and occur
in no tuple of $r$.

Now, given a consistent update of $r$, we construct a consistent update of $r^*$.
Let $r'$ be a consistent update of $r$ with $\distu(r',r)\le K$. We leave every
enforcing group unchanged and assign to the tuples $t^*$ the target triples of $r'$,
each extended by values that occur nowhere in $r^*$ and are distinct across clusters.
Every $t^*$ then changes its $c-3$ private cells in addition to the cells that $t$
changes in $r'$, at a total cost of at most $K+n(c-3)$. The result is consistent: the
extended target tuples pairwise disagree on $A_1$, $A_2$ and $A_3$, as do the target
triples of $r'$, and they agree with no enforcing tuple in any attribute.

Next, given a consistent update of $r^*$, we construct a consistent update of $r$.
Let $\hat{r}$ be a consistent update of $r^*$ with $\distu(\hat{r},r^*)\le K+n(c-3)$.
Since $N$ exceeds this budget, every enforcing group retains an unchanged copy, and so
every enforcing tuple $e_p$ is a target tuple of $\hat{r}$. As the target tuples of
distinct clusters disagree on every attribute, we conclude the following.
\begin{itemize}
  \item[($\star$)] For every private value $p$, the only target tuple of $\hat{r}$ that
        contains $p$ in the attribute in which $e_p$ contains it is $e_p$ itself.
\end{itemize}
We say that a tuple $t^*$ is \e{captured} if its target tuple is an enforcing tuple
$e_p$. Since $e_p$ contains, in every attribute other than the one that holds $p$, a
value that occurs nowhere else, a captured tuple retains at most one cell and therefore
costs at least $c-1$.

We first show that we may assume that no tuple of $\hat{r}$ is captured. Let $t^*$ be a
captured tuple. If some target tuple $v$ of $\hat{r}$ that is not an enforcing tuple
agrees with $t^*$ on an attribute, then we move $t^*$ to the cluster of $v$; the tuple
$t^*$ then retains a cell and costs at most $c-1$. Otherwise, we place $t^*$ in a
cluster of its own, the target tuple of which agrees with $t^*$ on $A_1$, $A_2$ and
$A_3$ and contains, in each of $A_4,\dots,A_c$, a value that occurs nowhere in $r^*$
and in no other target tuple; the tuple $t^*$ then costs $c-3$. In both cases, the
result is again a consistent update of $r^*$, and its cost is no higher than that of
$\hat{r}$. This is clear in the first case. In the second case, the new target tuple
contains new values in $A_4,\dots,A_c$ and agrees with $t^*$ on $A_1$, $A_2$ and $A_3$,
so it suffices to observe that no other target tuple agrees with $t^*$ on $A_1$, $A_2$
or $A_3$. For a target tuple that is not an enforcing tuple, this is the assumption of
the case; and an enforcing tuple contains in $A_1$, $A_2$ and $A_3$ values that occur
nowhere else in $r^*$, whereas $t^*$ contains there values of $r$. Each application of
this transformation reduces the number of captured tuples by one, so we may indeed
assume that $\hat{r}$ has none.

By ($\star$), a tuple $t^*$ that is not captured retains none of its private cells, and
hence it changes all $c-3$ of them, in addition to the cells among $A_1$, $A_2$ and
$A_3$ on which it differs from its target tuple. Denoting by $d_t$ the number of the
latter, we obtain that $\distu(\hat{r},r^*)\ge n(c-3)+\sum_{t\in r}d_t$ and, therefore,
that $\sum_{t\in r}d_t\le K$.

We construct $r'$ by assigning to every tuple $t$ of $r$ the restriction of the target
tuple of $t^*$ to $A_1$, $A_2$ and $A_3$. Every tuple then changes precisely $d_t$
cells, so $\distu(r',r)=\sum_{t\in r}d_t\le K$. Finally, $r'$ is consistent: two tuples
of $r'$ that originate in the same cluster of $\hat{r}$ are equal, and two tuples that
originate in distinct clusters disagree on each of $A_1$, $A_2$ and $A_3$, since the
target tuples of distinct clusters disagree on every attribute.
\end{proof}

\begin{reptheorem}{\ref{thm:kc-hard}}
For every $c \ge 3$, computing an optimal U-repair for $\scdepset{c}$ is NP-hard.
\end{reptheorem}

\begin{proof}
For $c=3$, this is Lemma~\ref{lem:k3-hard}, the proof of which shows that it is
NP-complete to decide whether a relation over $A_1$, $A_2$ and $A_3$ has a consistent
update of cost at most a given budget. For $c>3$, Lemma~\ref{lem:k3-to-kc} reduces this
decision problem, in polynomial time, to the corresponding problem for $\scdepset{c}$.
\end{proof}

\section{Missing Proofs of Section~\ref{sec:small-sccs}}\label{sec:app-small-scc}

In this section, we provide the missing proofs for Section~\ref{sec:small-sccs}.

\subsection{Proof of Proposition~\ref{prop:match-to-att-hard}}

\begin{repproposition}{\ref{prop:match-to-att-hard}}
Computing an optimal U-repair for $A\lra B\ra C$ is NP-hard.
\end{repproposition}

\begin{proof}
We reduce from \mathsc{Vertex Cover}. Let $G=(V,E)$ be a graph and let $k$ be a bound
on the size of a vertex cover; we may assume that $k\le|V|$, since $V$ itself is a
vertex cover. We set $L\defeq|V|+1$ and $M\defeq L|E|+|V|+1$.

We construct a relation $r$ over $\set{A,B,C}$, writing a tuple as a triple $(a,b,c)$
that lists its values for $A$, $B$ and $C$. The relation consists of the following
tuples.
\begin{itemize}
  \item \e{Anchors.} For every vertex $v\in V$, the relation contains the \e{anchors}
        of $v$, namely $M+1$ copies of $(v,v,0)$ and $M$ copies of $(v,v,1)$.
  \item \e{Edge gadgets.} For every edge $\set{u,v}\in E$, the relation contains the
        \e{edge gadget} of $\set{u,v}$, namely $L$ copies of $(u,v,1)$, where the two
        endpoints are ordered arbitrarily.
\end{itemize}
The budget is $K_r\defeq M|V|+L|E|+k$. We show that $G$ has a vertex cover of size at
most $k$ if and only if $r$ has a consistent update $r'$ with $\distu(r',r)\le K_r$.

Given a vertex cover of $G$, we construct a consistent update of $r$.
Let $S$ be a vertex cover of $G$ with $|S|\le k$, and let $c_v\defeq1$ if $v\in S$ and
$c_v\defeq0$ otherwise. We change to $c_v$ the $C$-value of every anchor of $v$ whose
$C$-value differs from $c_v$, at a cost of $M+1$ if $v\in S$ and of $M$ otherwise, and
therefore at a total cost of $M|V|+|S|$. For every edge $\set{u,v}$, we pick an
endpoint in $S$, say $u$, and change to $u$ the $B$-value of the $L$ copies of
$(u,v,1)$, at a cost of $L$ per edge. The total cost is $M|V|+|S|+L|E|\le K_r$.

Every tuple of the resulting relation is of the form $(v,v,c_v)$ for a vertex $v$: the
anchors of $v$ by construction, and each edge tuple because $c_u=1$ for the chosen
endpoint $u$. The FDs $A\ra B$ and $B\ra A$ hold because the $A$- and $B$-values of
every tuple coincide, and $B\ra C$ holds because all the tuples with $B=v$ have the
$C$-value $c_v$.

Next, given a consistent update of $r$, we construct a vertex cover of $G$.
Let $r'$ be a consistent update of $r$ with $\distu(r',r)\le K_r$. Since
$r'\models A\lra B$, the pairs of $A$- and $B$-values of $r'$ form a partial bijection:
all the tuples that share their $A$-value also share their $B$-value, and vice versa.
Moreover, since $r'\models B\ra C$, all the tuples that share their $B$-value share
their $C$-value as well.

Consider the anchors of a vertex $v$. In $r$, all of them use the pair $(v,v)$ in $A$
and $B$, and they split into a group of $M+1$ tuples with $C=0$ and a group of $M$
tuples with $C=1$, violating $B\ra C$. Hence, we must modify at least one cell in every
tuple of one of the two groups, at a cost of at least $M$ per vertex. The budget is
$M|V|+L|E|+k$, so this leaves a slack of $L|E|+k$.

Pulling the two groups apart is expensive. We say that an anchor of $v$ \e{stays} if it
retains both its $A$-value $v$ and its $B$-value $v$ in $r'$, and that it \e{leaves}
otherwise. Suppose that at least one anchor of $v$ stays and at least one leaves. The
anchors that stay use the pair $(v,v)$, so every tuple of $r'$ with $A=v$ has $B=v$
and every tuple of $r'$ with $B=v$ has $A=v$. Therefore, an anchor that leaves cannot
retain just one of $A=v$ and $B=v$: changing one of the two alone collides with the
group that stays. It must change both, at a cost of $2$ per leaving tuple. As the
smaller group has $M$ tuples, pulling a group apart costs at least $2M$. If no anchor
of $v$ stays, then each of the $2M+1$ anchors of $v$ changes a cell, which costs at
least $2M+1$. As $M$ exceeds the slack, neither option is affordable. Hence, every
vertex has an anchor that stays, and every vertex unifies its anchors on a single value
$c_v$ of $C$: the anchors that stay share the $B$-value $v$, hence they share their
$C$-value as well, and we denote it by $c_v$. An anchor that stays costs at
least $1$ if its $C$-value in $r$ differs from $c_v$, so the anchors of $v$ cost at
least $M$ if $c_v=0$, at least $M+1$ if $c_v=1$, and at least $2M+1$ if
$c_v\notin\set{0,1}$. The last option exceeds the slack, so $c_v\in\set{0,1}$. Let
$S\defeq\setof{v\in V}{c_v=1}$; the anchors cost at least $M|V|+|S|$.

The edge gadgets test whether $S$ covers $E$. Since an anchor of $v$ stays and carries
the pair $(v,v)$, the FD $A\ra B$ implies that every tuple of $r'$ with $A=v$ has
$B=v$, and $B\ra C$ implies that it has $C=c_v$; symmetrically, every tuple with $B=v$
has $A=v$ and $C=c_v$. Now consider the gadget of an edge $\set{u,v}$. None of its $L$
copies of $(u,v,1)$ can retain both $A=u$ and $B=v$, so each of them must join the
cluster of $u$ or the cluster of $v$, or leave both. Joining that of $u$ amounts to
changing the $B$-value to $u$, and it changes nothing else exactly when $c_u=1$, since
then $C=1$ is already the $C$-value of the cluster; if $c_u=0$, the copy must also
change its $C$-value to $0$. The argument for joining the cluster of $v$ is symmetric.
A copy that joins neither cluster changes both its $A$-value and its $B$-value. Hence,
the gadget costs at least $L$, and at least $2L$ if $c_u=c_v=0$, that is, if the edge
is not covered by $S$.

Let $E_0$ be the set of the edges that $S$ does not cover. Putting the bounds together,
\[
  K_r\;\ge\;\distu(r',r)\;\ge\;M|V|+|S|+L|E|+L|E_0| ,
  \qquad\text{hence}\qquad
  |S|+L|E_0|\;\le\;k .
\]
Since $L=|V|+1>k$, we get $E_0=\emptyset$, so $S$ is a vertex cover, and $|S|\le k$.
\end{proof}

\subsection{Proof of Proposition~\ref{prop:att-to-match-hard}}

\begin{repproposition}{\ref{prop:att-to-match-hard}}
Computing an optimal U-repair for $A\ra B\lra C$ is NP-hard.
\end{repproposition}

\begin{proof}
We reduce from \mathsc{3SAT}, assuming that every clause consists of three literals
over three distinct variables. Let $\phi$ be a formula with variable set $V$ and $m$
clauses, and set $L\defeq 4m+1$. With every variable $v$ we associate two fresh values
$T_v$ and $F_v$, one for each truth value. The constructed relation over $\set{A,B,C}$
consists of: 
\begin{itemize}
    \item a \e{variable gadget} for every variable $v$, namely $L$ copies of
$(v,v,T_v)$ and $L$ copies of $(v,v,F_v)$, and 
\item a \e{clause gadget} for every clause
$c=(\ell_1\vee\ell_2\vee\ell_3)$ over the variables $v_1,v_2,v_3$, namely the three
tuples: 
\[
  (c,v_1,w_3)\qquad (c,v_2,w_1) \qquad (c,v_3,w_2)
\]
where $w_i$ is the value of $v_i$ that satisfies $\ell_i$, that is, $T_{v_i}$ if
$\ell_i$ is positive and $F_{v_i}$ if $\ell_i$ is negative.
\end{itemize}
Note that the indices in a clause gadget
are shifted: $v_1$, $v_2$ and $v_3$ coincide with $w_3$, $w_1$ and $w_2$, respectively. For example, if $c$ is $x\lor y\lor \neg z$, then the three tuples are:
\[
  (c,x,F_z)\qquad (c,y,T_x) \qquad (c,z,T_y)
\]
The budget is $K\defeq L|V|+4m$. 

Given a satisfying assignment $\tau$ of $\phi$, we construct a consistent update of
cost exactly $K$. Denote by $\sigma(v)$ the value $T_v$ if $\tau(v)=\true$ and the
value $F_v$ otherwise. 
\begin{itemize}
\item For every variable $v$, we change to $\sigma(v)$ the $C$-value
of the $L$ tuples of its gadget that carry the other value (cost $L|V|$ in total). 
\item For
every clause $c$, we pick a literal $\ell_i$ that $\tau$ satisfies, and we replace the
three tuples of its gadget by $(c,v_i,w_i)$; one of the three has already the $B$-value
$v_i$, another already has the $C$-value $w_i=\sigma(v_i)$, and the third has neither
(cost $1+1+2=4$ per clause, hence $4m$ in total). 
\end{itemize}
The resulting relation $r'$ is consistent: it satisfies $A\ra B$ since the
tuples that share an $A$-value are those of a single gadget, and it satisfies $B\lra C$
since its pairs of $B$- and $C$-values are precisely the pairs $(v,\sigma(v))$.

Now suppose that we are given a consistent update $r'$ of cost at most $K$. Since $B\lra C$,
the pairs of $B$-values and $C$-values that occur in it form a partial bijection, which we denote by $\sigma$. If no
tuple has the $B$-value $v$, then $\sigma(v)$ is undefined, and we let it be a value that is neither $T_v$ nor $F_v$, say $\diamond$. Hence, $\sigma(v)$ is defined for every variable $v$. 

Consider the gadget of a variable $v$, and let $t$ be a tuple of the
gadget whose $C$-value in $r$ differs from $\sigma(v)$. The update cannot keep both the
$B$-value and the $C$-value of $t$, and it therefore changes at least one cell of $t$. Thus, the gadget costs at least the number of its tuples whose $C$-value differs from
$\sigma(v)$. As the gadget splits evenly between $T_v$ and $F_v$, this number is at least $L$ if
$\sigma(v)$ is one of $T_v$ and $F_v$, and at least $2L$ if it is a different $C$-value in $r'$. The cost is also at least $2L$ if $\sigma(v)=\diamond$, since it means that all $2L$ occurrences of $v$ had to disappear in the update. 

In conclusion, the variable gadgets cost at least $L|V|$ in total, and at least
$L|V|+L$ if $\sigma$ maps some variable to neither of its two values. Since $L$ exceeds
the remaining budget of $4m$, the latter is impossible. Hence, $\sigma$ maps every variable $v$ to one of $T_v$ and $F_v$, and we can treat it as a truth assignment $\tau$. We complete the proof by showing that $\tau$ is a satisfying assignment. Since the remaining budget is $4m$, it suffices to show that a clause gadget costs at least $4$, and that it
costs exactly $4$ only if $\tau$ satisfies its clause. We do so next.

Due to the shift, every tuple of
a clause gadget pairs a variable with a value of a \e{different} variable; no such pair belongs to $\sigma$, so each of the three tuples changes its $B$-value or its
$C$-value, and the gadget costs at least $3$. If the three tuples do not all retain the $A$-value $c$, then at least one of them changes its $A$-value as well, and the gadget
costs at least $4$. If they all retain it, then $A$ determines both $B$ and $C$, so the
three tuples collapse into a single pair of $\sigma$; their three $B$-values are
distinct and so are their three $C$-values, hence this pair agrees with at most two of
the six $B$-cell and $C$-cells of the gadget, and the gadget again costs at least $4$. In both cases, a cost of exactly $4$ forces a pair of $(v,\sigma(v))$ that agrees with one $B$-cell and one $C$-cell of the gadget:
\begin{itemize}
\item In the first case, exactly one tuple changes its $A$-value and one of its $B$ of $C$ values, so the two tuples that retain their $A$ have a budget of 2. Since they agree on $A$, they should be updated into identical tuples, and the only way we can do so is by retaining at least one of the $w_i$ of the gadget, and $w_i$ satisfies the clause (by construction). 
\item In the second case, where $A$-values are retained, the three tuples should become identical, and the gadget should update at least two $B$-cells and at least two $C$-cells. Hence, it updates precisely two $C$-cells, and again at least one $w_i$ is retained and satisfies the clause.
\end{itemize}
We conclude that the clause is necessarily satisfied and $\tau$ is a satisfying assignment, as claimed.
\end{proof}

\subsection{Proof of Proposition~\ref{prop:match-to-match-hard}}

\begin{repproposition}{\ref{prop:match-to-match-hard}}
Computing an optimal U-repair for $A\lra B\ra C\lra D$ is NP-hard.
\end{repproposition}

\begin{proof}
We reduce from $A\lra B\ra C$, which is intractable by
Proposition~\ref{prop:match-to-att-hard} and consists of the first three attributes of
$A\lra B\ra C\lra D$. In the spirit of Lemma~\ref{lem:k3-to-kc}, the reduction pads
every tuple with a private value in the new attribute $D$ and anchors this value by an
enforcing group.

Let $r$ be a relation over $\set{A,B,C}$ with $n$ tuples, let $K$ be a budget, and set
$N\defeq K+n+1$. For every $i\in\tids(r)$, the relation $r'$ over $\set{A,B,C,D}$
contains:
\begin{itemize}
    \item the \e{main tuple} $M_i$ that agrees with $r[i]$ on $A$, $B$ and $C$ and
contains a \e{private} value $d_i$ in $D$, and
\item for every private value $d_i$,
the relation $r'$ contains an \e{enforcing group} of $N$ copies of the \e{enforcing
tuple} $E_i$ that contains $d_i$ in $D$ and, in each of $A$, $B$ and $C$, a value that
occurs nowhere else. 
\end{itemize}
The
budget is $K+n$. We show that $r$ has a consistent update of cost at most $K$ with
respect to $A\lra B\ra C$ if and only if $r'$ has a consistent update of cost at most
$K+n$ with respect to $A\lra B\ra C\lra D$.

Now, given a consistent update of $r$, we construct a consistent update of $r'$,
Let $s$ be a consistent update of $r$ with $\distu(s,r)\le K$. We may assume that $s$
uses none of the values that the construction introduces, since renaming the values of
$s$ that do not occur in $r$ affects neither its cost nor its consistency. With every
value $c$ that occurs in the $C$-column of $s$ we associate a fresh value $\delta_c$,
and we define the update $s'$ of $r'$ that leaves every enforcing group unchanged and
replaces $M_i$ by the tuple $(s[i][A],s[i][B],s[i][C],\delta_{s[i][C]})$. Every main
tuple changes its $D$-cell, since $\delta_{s[i][C]}$ is fresh and $d_i$ is not, and the
main tuples change $\distu(s,r)$ cells in $A$, $B$ and $C$ in total; hence
$\distu(s',r')=\distu(s,r)+n\le K+n$.

We verify that $s'\models A\lra B\ra C\lra D$. Among the main tuples, $A\lra B$ and
$B\ra C$ hold because $s\models A\lra B\ra C$, and $C\lra D$ holds because
$c\mapsto\delta_c$ is injective. An enforcing tuple $E_i$ shares no value with any
other tuple of $s'$ in any attribute: its $A$-, $B$- and $C$-values occur only in its
own enforcing group, and its private value $d_i$ is no longer used by $M_i$. Hence, no
FD is violated.

Next, given a consistent update of $r'$, we construct a consistent update of $r$.
Let $s'$ be a consistent update of $r'$ with $\distu(s',r')\le K+n$. Since $N$ exceeds
this budget, every enforcing group retains an unchanged copy, and so every enforcing
tuple $E_i$ occurs in $s'$. As the enforcing tuples occur in $s'$, the FDs $C\ra D$ and
$D\ra C$ imply that a main tuple agrees with an enforcing tuple on $C$ if and only if
it agrees with it on $D$.

We first show that we may assume that no main tuple agrees with an enforcing tuple on
$C$. Consider a group $Q$ of main tuples that have been updated to the same
$(a,b,c,d)$, where $(c,d)$ is the pair of $C$- and $D$-values of an enforcing tuple
$E$. Every tuple of $Q$ has changed its $C$-cell, since $c$ occurs nowhere in $r$.
Moreover, at most one tuple of $Q$ retains its original $D$-value, namely the main
tuple whose private value is $d$, if it belongs to $Q$. Let $j\in Q$ and let $c'$ be
the original $C$-value of $j$. If $c'$ already occurs in $s'$, let $d'$ be the
$D$-value already paired with it; otherwise, let $d'$ be a fresh value. We replace the
pair $(c,d)$ of every tuple of $Q$ by $(c',d')$. If $(a,b)$ is the pair of $A$- and
$B$-values of $E$, those cells were already changed, since the values of $E$ occur
nowhere in $r$, and we replace them by a fresh pair as well, so as not to violate
$B\ra C$ against $E$. One $C$-cell is then restored, namely that of $j$, and at most
one extra $D$-cell is changed, namely if some tuple of $Q$ had retained $d$ and $d'$
differs from $d$. Hence, the cost does not increase.

The result is again a consistent update of $r'$. The tuples of $Q$ now agree on
$(c',d')$. If we replaced $A$ and $B$, they agree on a fresh pair that occurs nowhere
else. If we kept $(a,b)$, then no tuple outside $Q$ has the $B$-value $b$: the
enforcing tuple $E$ has a different $B$-value, and every other tuple with $B=b$ would
have had $C=c$ by $B\ra C$, hence would have belonged to $Q$. The value $c'$ occurs in
$r$ and is therefore the $C$-value of no enforcing tuple, so the tuples of $Q$ no
longer agree with an enforcing tuple on $C$ or on $D$. Repeating the modification for
every such group, we obtain a consistent update of $r'$ of cost at most $K+n$ in which
no main tuple agrees with an enforcing tuple on $D$, and we assume from now on that
$s'$ is such an update.

In $s'$, no main tuple retains its private value, since it would then agree with its
own enforcing tuple on $D$ and hence on $C$. The main tuples therefore change $n$ cells
in the attribute $D$, and so they change at most $K$ cells in $A$, $B$ and $C$. Let $s$
be the update of $r$ obtained by restricting the main tuples of $s'$ to $A$, $B$ and
$C$; then $\distu(s,r)\le K$, and $s$ is consistent, since the FDs $A\ra B$, $B\ra A$
and $B\ra C$ are implied by $A\lra B\ra C\lra D$ and hence hold in $s'$, in particular
among the main tuples.
\end{proof}

\end{document}